\documentclass{IEEEtran}

\usepackage{pstricks-add,pst-node,pst-tree}
\usepackage{epsfig}
\usepackage{pst-grad} 
\usepackage{pst-plot} 
\usepackage[space]{grffile} 
\usepackage{etoolbox} 
\makeatletter 
\patchcmd\Gread@eps{\@inputcheck#1 }{\@inputcheck"#1"\relax}{}{}
\makeatother

\usepackage{mathdots}
\usepackage[english]{babel}
\usepackage{relsize}
\usepackage{graphics,graphicx}
\usepackage{url}

\usepackage{booktabs}
\usepackage{multirow}
\usepackage{mparhack}
\usepackage{subfigure}
\usepackage{authblk}
\usepackage{amsthm}
\usepackage{mathrsfs}
\usepackage{hyperref}

\usepackage{array}
\usepackage{cite}
\usepackage{eqparbox}
\usepackage{mdwmath}

\usepackage{epsfig}
\usepackage{xcolor}

\usepackage{mathtools,cuted}
\usepackage{bm}
\usepackage{bbold}
\usepackage[all]{xy}
\usepackage{etoolbox}

\usepackage[euler-digits,euler-hat-accent]{eulervm} 
                               
\usepackage{clrscode3e}

\DeclareMathAlphabet{\mathantt}{OT1}{antt}{li}{it}
\DeclareMathAlphabet{\mathpzc}{OT1}{pzc}{m}{it}

\DeclarePairedDelimiter\norm{\lVert}{\rVert}%

\usepackage{accents}
\usepackage{multicol}

\usepackage{lipsum,color}

\usepackage{tcolorbox}

\usepackage{cite}

\usepackage{pgfplots}
\usepackage{tikz}
\usetikzlibrary{plotmarks}
\usetikzlibrary{patterns}
\usetikzlibrary{shapes.geometric}

\newtheorem{theorem}{Theorem}

\newtheorem{lemma}{Lemma}

\newtheorem{definition}{Definition}

\def\I{\mathcal{C}}
\def\J{\mathcal{J}}

\def\C{\mathcal{C}}
\def\G{\mathcal{G}}
\def\u{\mathpzc{u}}
\def\U{\mathcal{U}}

\def\V{\mathcal{V}}
\def\Y{\mathcal{Y}}

\def\minf{\textsc{MinF}}

\DeclareMathOperator{\argmin}{\arg\min}

\renewcommand{\vec}[1]{\mathbold{#1}}
\renewcommand{\bm}[1]{\mathbold{#1}}

\usepackage{xparse}
\NewDocumentCommand{\overarrow}{O{=} O{\uparrow} m}{%
  \overset{\makebox[0pt]{\begin{tabular}{@{}c@{}}#3\\[0pt]\ensuremath{#2}\end{tabular}}}{#1}
}
\NewDocumentCommand{\underarrow}{O{=} O{\downarrow} m}{%
  \underset{\makebox[0pt]{\begin{tabular}{@{}c@{}}\ensuremath{#2}\\[0pt]#3\end{tabular}}}{#1}
}

\begin{document}

\title{\textcolor[rgb]{0,0,0}{Resource Optimization with Load Coupling \\in Multi-cell NOMA}}
\author{

    \IEEEauthorblockN{Lei~You\IEEEauthorrefmark{1}, Di~Yuan\IEEEauthorrefmark{1}, Lei~Lei\IEEEauthorrefmark{2}, Sumei~Sun\IEEEauthorrefmark{3}, Symeon~Chatzinotas\IEEEauthorrefmark{2}, and Bj\"{o}rn~Ottersten\IEEEauthorrefmark{2}
    }\\
    
    \IEEEauthorblockA{\IEEEauthorrefmark{1}Department of Information Technology, Uppsala University, Sweden
    \\\{lei.you; di.yuan\}@it.uu.se
    }
    
    \IEEEauthorblockA{\IEEEauthorrefmark{2}Interdisciplinary Centre for Security, Reliability and Trust, Luxembourg University, Luxembourg
    \\\{lei.lei; symeon.chatzinotas; bjorn.ottersten\}@uni.lu
    }
    
    \IEEEauthorblockA{\IEEEauthorrefmark{3}Institute for Infocomm Research, A*STAR, Singapore
    \\ sunsm@i2r.a-star.edu.sg
    }
    
    \thanks{Part of this paper has been presented at IEEE GLOBECOM, Singapore, Dec. 2017 \cite{lei:globecom17}.}

}



\maketitle


\begin{abstract}
Optimizing non-orthogonal multiple access (NOMA) in multi-cell scenarios is much more challenging than the single-cell case because inter-cell interference must be considered. 
Most papers addressing NOMA consider a single cell. 
We take a significant step of analyzing NOMA in multi-cell scenarios. \textcolor[rgb]{0,0,0}{We explore the potential of NOMA networks in achieving optimal resource utilization with arbitrary topologies. Towards this goal, we investigate a broad class of problems consisting in optimizing power allocation and user pairing for any cost function that is monotonically increasing in time-frequency resource consumption. We propose an algorithm that achieves global optimality for this problem class. The basic idea is to prove that solving the joint optimization problem of power allocation, user pair selection, and time-frequency resource allocation amounts to solving a so-called iterated function without a closed form. We prove that the algorithm approaches optimality with fast convergence.} Numerically, 
we evaluate and demonstrate the performance of NOMA for multi-cell scenarios in terms of resource efficiency and load balancing.
\end{abstract}
\begin{IEEEkeywords}
NOMA, multi-cell, resource allocation
\end{IEEEkeywords}

\section{Introduction}

\IEEEPARstart{N}{on-orthogonal} multiple access (NOMA) is considered as a promising technique for enhancing resource efficiency~\cite{7676258,2016arXiv161101607S,7273963,R1-153332,7582424,7357604,DBLP:journals/corr/TabassumHH16,7964738,7878674,7974737,7973138,7982784}.
In two recent surveys~\cite{7676258,2016arXiv161101607S}, the authors pointed out that resource allocation in multi-cell NOMA poses much more research challenges compared to the single-cell case, because optimizing NOMA with multiple cells has to model the interplay between successive interference cancellation (SIC) and inter-cell interference. As one step forward, the investigations in \cite{7676258,2016arXiv161101607S} have addressed two-cell scenarios. 
In~\cite{7582424}, the authors proposed two coordinated NOMA beamforming methods for two-cell scenarios. 
Reference~\cite{DBLP:journals/corr/TabassumHH16} uses stochastic geometry to model the inter-cell interference in NOMA. Hence the results do not apply for analyzing network with specific given network topology.
Reference~\cite{7964738} optimizes energy efficiency in multi-cell NOMA with downlink power control. However, 
the aspect of determining which users share resource by SIC, i.e., \textit{user pairing}, is not considered.
To the best of our knowledge, finding \textit{optimal power allocation and user pairing} simultaneously for enhancing network resource efficiency in multi-cell NOMA without restrictions on network topology has not been addressed yet. 

The crucial aspect of multi-cell NOMA consists of capturing the mutual interference among cells; This is a key consideration in SIC of NOMA. Therefore, the cells cannot be optimized independently.
For orthogonal multiple access (OMA) networks, a modeling approach had been proposed that characterizes the inter-cell interference via capturing the mutual influence among the cells' resource allocations \textcolor[rgb]{0,0,0}{
\cite{7959870,7132788,5198628,6204009,6732895,7585124,7880696,5450287,5489842,6292896,6363999,6479364,6747283,6924853,6815652,7151124,6887352,7480379,7273956,7332797,7962728,7744690,DBLP:journals/corr/abs-1710-09318}}. The model, named~\textit{load-coupling}, refers to the time-frequency resource consumption in each cell as the \textit{cell load}. However, 
the model does not allow SIC.
In our recent work~\cite{lei:globecom17}, we addressed resource optimization in multi-cell NOMA.
However, the system model is constrained by fixed power allocation. How to model joint optimization of power allocation and user pairing and how to solve the resulting problems to optimality have remained open so far.

\section{Main Results}

\textcolor[rgb]{0,0,0}{Thus far, for multi-cell NOMA, stochastic geometry is adopted to model inter-cell interference \cite{DBLP:journals/corr/TabassumHH16}, which results in difficulties for analysis upon specific network topologies. In this paper, we target optimizing multi-cell NOMA network with any given topology. In the modeling approaches of OMA used by \cite{7959870,7132788,5198628,6204009,6732895,7585124,7880696,5450287,5489842,6292896,6363999,6479364,6747283,6924853,6815652,7151124,6887352,7480379,7273956,7332797,7962728,7744690,DBLP:journals/corr/abs-1710-09318}, instead of making micro-level assumptions on the behavior of the resource scheduler or slot-by-slot consideration of inter-cell interference per resource block (RB) in each individual cell, the level of interference generated by a cell is directly related to the amount of allocated time-frequency resource in the cell. This is used to model the coupling relationship of resource allocation among cells, which is shown to be sufficiently accurate for network-level interference characterization \cite{6363999}, \cite{7273956}.}

We demonstrate how NOMA can be modeled in multi-cell scenarios by significantly extending the approaches in~\textcolor[rgb]{0,0,0}{\cite{7959870,7132788,5198628,6204009,6732895,7585124,7880696,5450287,5489842,6292896,6363999,6479364,6747283,6924853,6815652,7151124,6887352,7480379,7273956,7332797,7962728,7744690,DBLP:journals/corr/abs-1710-09318}}, with joint optimization of power allocation and user pairing. 
\textcolor[rgb]{0,0,0}{One fundamental result under such type of models in OMA is the existence of the equilibrium for resource allocation. However, this modeling approach in NOMA leads to non-closed form formulation of cell load coupling, unlike the case of OMA. The fact poses significant challenges in analysis and problem solving. As one of our main results, we prove that such an equilibrium for resource allocation in NOMA exists as well and propose an efficient algorithm for obtaining the equilibrium. Furthermore, we prove that the equilibrium is the global optimum for resource optimization in multi-cell NOMA and thus a wide class of resource optimization problems can be optimally solved by our algorithm. 
Because of our analytical results, previous works about OMA with load coupling is a special case of ours, namely, the algorithmic notions and mathematical tools being used in those works of classic multi-cell power control or OMA load coupling thus directly apply to the analysis multi-cell NOMA, suggesting future works on this topic.} \textcolor[rgb]{0,0,0}{All our analytical results are based on the extended model.}

To the best of our knowledge, this is the first work investigating how to \textit{optimally utilize power and time-frequency resources jointly} in multi-cell NOMA. As a key strength of our modeling approach, it enables to formulate and optimize an entire class of resource optimization problems. \textcolor[rgb]{0,0,0}{Namely, as long as the cost function is monotonically increasing in the cells' time-frequency resource consumption, our proposed framework in multi-cell NOMA applies.} Specifically, for solving this class of problems optimally, we derive a \textit{polynomial-time} algorithm \textsc{S-Cell} that gives the optimal power allocation and user pairing, for any given input of inter-cell interference. To address the dynamic coupling of inter-cell interference, we derive a unified algorithmic framework \textsc{M-Cell} that solves the multi-cell resource optimization problems optimally. The algorithm \textsc{S-Cell} serves as a sub-routine and is iteratively called by \textsc{M-Cell}. We demonstrate theoretically the \textit{linear convergence} of this process.

\textcolor[rgb]{0,0,0}{The fundamental differences between our investigated problems and single-cell NOMA are summarized as follows. For multiple cells, the resource allocation in one cell affects the interference that the cell generates to other cells. The amounts of required resource to meet the demand for all cells are coupled together, rather than being independent to each other. Optimizing resource allocation within one cell leads to a chain reaction among all other cells. Individual optimization for the cells results in sub-optimality and very inaccurate performance analysis. For multi-cell NOMA, not only the time-frequency resource allocation but also the power splits and user pairings in all cells are coupled together for the same reason. Therefore, joint optimization in NOMA leads to a rather complex problem for analysis.}

By numerical experiments, optimizing resource utilization by our algorithm
enlightens how much we can gain from NOMA in terms of \textit{resource efficiency} and \textit{load balancing}.

\section{Cell Load Modeling}
\label{sec:load}

Denote by $\I=\{1,2,\ldots,n\}$ and $\J=\{1,2,\ldots,m\}$ the sets of cells and user equipments (UEs), respectively. Denote by $\J_i$ the set of UEs served by cell $i$ ($i\in\I$). When using $j$ to refer to one UE in $\J$, $i$ by default indicates $j$'s serving cell, unless stated otherwise. 
Downlink is considered in our model.

\subsection{Resource Sharing in NOMA} 

The resource in time-frequency domain is divided into RBs. In OMA, one RB can be accessed by only one UE. In NOMA, multiple UEs can be clustered together to access the same RB by SIC.
Increasing the number of UEs in SIC, however, leads to fast growing decoding complexity~\cite{7676258,2016arXiv161101607S}. In previous works, it has been
demonstrated that most of the possible performance improvement by SIC is reached by pairing as few as two UEs \cite{7676258,2016arXiv161101607S,7273963,R1-153332,7582424,7357604}. Pairing two UEs for resource sharing is illustrated in~\figurename~\ref{fig:RB}.
\begin{figure}[!t]
\centering
\includegraphics[width=\linewidth]{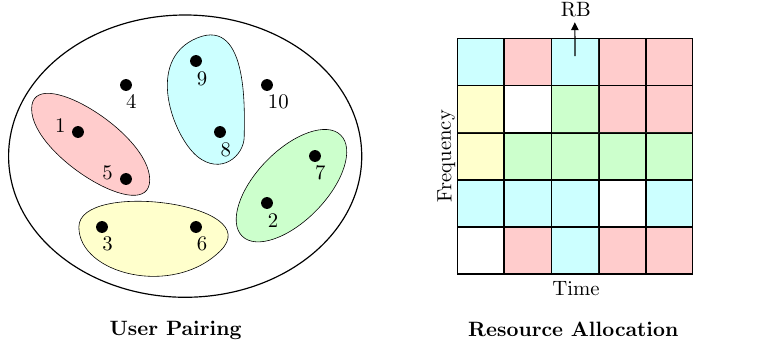}
\caption{This figure illustrates user pairing and time-frequency resource sharing. There are $10$ UEs in one cell. Eight form four user pairs $\{1,5\}$, $\{2,7\}$, $\{3,6\}$, $\{8,9\}$, and the other two UEs $4$ and $10$ are unpaired. The UEs within one pair share the same time-frequency resource as indicated by the colors.}
\label{fig:RB}
\end{figure}
UEs within one pair share the same RB and the RBs allocated to different pairs do not overlap.
We use $\mathpzc{u}$ as a generic notation for a user pair (referred to as ``pair'' for simplicity). 
For cell $i$ ($i\in\I$), denote by $\U_i$ the set of candidate pairs. Suppose there are in total $m_i$ UEs in cell $i$. Then $|\U_i|$ is up to ${m_i \choose 2}$. Denote by $\V_j$ ($j\in\J$) the set of pairs containing UE $j$. Let $\U=\bigcup_{i\in\I}\U_{i}$ (or equivalently $\U=\bigcup_{j\in\J}\V_j$) be the set of candidate pairs of all cells. Let $s=|\U|$. If there is a need to differentiate between pairs, we put indices on $\u$, i.e., $\U=\{\u_1,\u_2,\u_3,\ldots \u_s\}$. Finally, in our model, for UEs we allow for both OMA and NOMA with SIC. For each UE, that which mode is used (or both can be used) is determined by optimization. In the following, we refer to these two modes as   \textit{orthogonal RB allocation} and \textit{non-orthogonal RB allocation}, respectively. In general NOMA, we include both modes.

\subsection{NOMA Downlink}

We first consider orthogonal RB allocation in NOMA.
Let $p_i$ be the transmission power per RB in cell $i$ ($i\in\I$). Denote by $g_{ij}$ the \textcolor[rgb]{0,0,0}{channel coefficient} from cell $i$ to UE $j$. 
The signal-to-interference-and-noise ratio (SINR) is:
\begin{equation}
\gamma_j = \frac{p_{i}g_{ij}}{\sum_{k\in\I\backslash\{i\}}
   I_{kj}+\sigma^2}.
\label{eq:sinrOMA}
\end{equation}
\textcolor[rgb]{0,0,0}{The term $I_{kj}$ denotes the inter-cell interference from cell $k$ to UE $j$, and is possibly zero. This generic notation is used for the sake of presentation. Later, we use the load-coupling model, where the cell load that reflects the usage of RBs governs the amount of interference.} The term $\sigma^2$ is the noise power.

We then consider non-orthogonal RB allocation in NOMA. In~\cite{tse2005fundamentals} (Chapter 6.2.2, pp. 238) it is shown that, with superposition coding, one UE of pair $\u$ ($\u\in\U$) can decode the other by SIC.
When there is need to consider the decoding order in $\u$, to be intuitive, we use $\oplus$ to denote the UE that applies interference cancellation, followed by decoding its own signal. And $\ominus$ denotes the UE that only decodes its own signal. \textcolor[rgb]{0,0,0}{Note that both $\oplus$ and $\ominus$ are generic notations and refer to the two different users in any pair $\u$ ($\u\in\U$) in consideration.}
For any pair $\u$, $p_i$ is divided to $q_{\oplus\u}$ and $q_{\ominus\u}$ ($q_{\oplus\u}+q_{\ominus\u}=p_i$), with $q_{\oplus\u}$ and $q_{\ominus\u}$ being allocated to $\oplus$ and $\ominus$, respectively. (The generic notation $q_{j\u}$ ($j\in\u$) denotes the power allocated to UE $j$.) We remark that $\oplus$ decodes $\ominus$'s signal first and hence $\ominus$'s signal does not compose the interference for $\oplus$. The SINR of $\oplus$ is computed by~\eqref{eq:sinr+}.
\begin{equation}
\gamma_{\oplus\u} = \frac{q_{\oplus\u}g_{i\oplus}}{\sum_{k\in\I\backslash\{i\}}
   I_{k\oplus}+\sigma^2}.
\label{eq:sinr+}
\end{equation}
The UE $\ominus$ is subject to intra-cell interference from $\oplus$, i.e., 
\begin{equation}
\gamma_{\ominus\u} = \frac{q_{\ominus\u}g_{i\ominus}}
{\underbrace{
q_{\oplus\u}g_{i\ominus}}_{\textnormal{intra-cell}} 
+\underbrace{\sum_{k\in\I\backslash\{i\}}
   I_{k\ominus}
}_{\textnormal{inter-cell}}
      +\sigma^2
}.
\label{eq:sinr-}	
\end{equation}

Denote by $\vec{q}$ the power allocation of all candidate pairs:
\[
\vec{q}=\left[
\begin{matrix}
q_{\oplus\u_1}  & q_{\oplus\u_2}  & \cdots & q_{\oplus\u_s}\\
q_{\ominus\u_1} & q_{\ominus\u_2} & \cdots & q_{\ominus\u_s}
\end{matrix}\right].
\]
We use $\vec{q}_{\u}$ to represent the column of pair $\u$ ($\u\in\U$) in $\vec{q}$, named \textit{power split} for $\u$. 
We remark that it is not necessary to use all the pairs in $\U$ for resource sharing. Whether or not a pair would be put in use and allocated with RBs is determined by optimization, discussed later in Section~\ref{subsec:RB-alloc}. 
In addition, we remark that the decoding order is not constrained by the power split~\cite{tse2005fundamentals}, even though by our numerical results, more power is always allocated to $\ominus$ in optimal solutions. The issue of the influence of inter-cell interference on the decoding order is addressed later in Section~\ref{subsec:decoding}.

\subsection{Inter-cell Interference Modeling}
\label{subsec:averaging}

\textcolor[rgb]{0,0,0}{The basic idea is to use the cells' RB consumption levels to characterize respectively the cell's likelihood of interfering to the others. The approach is specified as follows.}
Denote by $\rho_{k}$ the proportion of RBs allocated for serving UEs in cell $k$. \textcolor[rgb]{0,0,0}{The intuition behind the model is partially explained by the two extreme cases $\rho_k=1$ and $\rho_k=0$.} If cell $k$ is fully loaded, meaning that all RBs are allocated, then $\rho_{k}=1$. In the other extreme case, cell $k$ is idle and accordingly $\rho_{k}=0$. Consider any UE $j$ served by cell $i$. The interference $j$ receives from cell $k$ is $I_{kj}=p_{k}g_{kj}$ or $I_{kj}=0$ in the two cases, respectively. 
In general, $\rho_k$ serves as a scaling parameter for interference, see~\eqref{eq:Ikj}. \textcolor[rgb]{0,0,0}{By the interference modeling approach, the cell load directly translates to the scaling effect of interference and therefore the same notation is used for both.}
\begin{equation}
I_{kj} = p_{k}g_{kj}\rho_{k}	.
\label{eq:Ikj}
\end{equation}

Intuitively, $\rho_{k}$ reflects the likelihood that a UE outside cell $k$ receives interference from $k$. Note that $\rho_k$ in fact is the amount of time-frequency resource consumption of cell $k$ and hence is referred to as \textit{the load of cell} $k$. 

\textcolor[rgb]{0,0,0}{
We remark that this type of interference modeling approach is a suitable approximation for network-level performance analysis, which enables study of inter-cell interference in large-scale multi-cell networks without having to modeling micro-level interference. Detailed system-level simulations (e.g. \cite{6363999} and \cite{7273956}) have shown that this type of modeling has sufficient accuracy for cell-level interference characterization. This approach has been widely used and is getting increasingly popular \cite{7959870,7132788,5198628,6204009,6732895,7585124,7880696,5450287,5489842,6292896,6363999,6479364,6747283,6924853,6815652,7151124,6887352,7480379,7273956,7332797,7962728,7744690,DBLP:journals/corr/abs-1710-09318}, which however, to our best knowledge, are limited to OMA. We provide analytical results in order to extend the modeling approach to NOMA. }

\begin{tcolorbox}[boxrule=0.5pt, colback=yellow!5]
\textbf{We remark that this section can be dropped without affecting any of the theoretical results in this paper. Please refer to our note in \href{https://arxiv.org/pdf/1909.08651.pdf}{\texttt{https://arxiv.org/pdf/1909.08651.pdf}} for more details.}
\subsection{Determining Decoding Order}
\label{subsec:decoding}
Inter-cell interference affects the decoding order in NOMA, and thus how to model the load-coupling in NOMA is significantly more challenging than OMA.
Lemma~\ref{rmk:decoding} below resolves this issue by identifying pairs for which the decoding order can be determined independently of interference. 
As another benefit, it significantly reduces the set of candidate pairs.
\begin{lemma}
\textcolor[rgb]{0,0,0}{For $\u=\{j,h\}$ ($g_{ij}\geq g_{ih}$) in cell $i$, if $g_{ij}/g_{ih}\geq g_{kj}/g_{kh}$ ($k\in\C\backslash\{i\}$), then SIC at $j$ decodes first the signal for $h$, followed by decoding its own signal, and, user $h$ does not apply SIC. That is, the decoding order $\oplus=j$ and $\ominus=h$ always hold for pair independent of interference.}
\label{rmk:decoding}
\end{lemma}
\begin{proof}
Denote by $\gamma_{hj}$ and $\gamma_{hh}$ respectively in  \eqref{eq:-+} and \eqref{eq:--} the SINRs at users $j$ and $h$ for the downlink signal of $h$.
\begin{equation}
    \gamma_{hj}=\frac{q_{h\u}g_{ij}}{q_{j\u}g_{ij}+
        \sum_{k\in\I\backslash\{i\}}
  p_{k}g_{kj}\rho_k
+\sigma^2}.
\label{eq:-+}
\end{equation}
\begin{equation}
    \gamma_{hh}=\frac{q_{h\u}g_{ih}}{q_{j\u}g_{ih}+
    \sum_{k\in\I\backslash\{i\}}
  p_{k}g_{kh}\rho_k+\sigma^2}.
  \label{eq:--}
\end{equation}
With superposition coding, $j$ cancels the interference from $h$ if $j$ can decode any data that $h$ can decode~\cite{tse2005fundamentals}, i.e. $\gamma_{hj}\geq\gamma_{hh}$, which reads:
\begin{multline*}
q_{j\u}g_{ij}g_{ih}+g_{ij}\sum_{k\in\I\backslash\{i\}}
    p_{k}g_{kh}\rho_{k}+g_{ij}\sigma^2 \\
    \geq q_{j\u}g_{ij}g_{ih}+g_{ih}\sum_{k\in\I\backslash\{i\}}p_{k}g_{kj}\rho_k+g_{ih}\sigma^2.
\end{multline*}
Further, $\gamma_{hj}\geq\gamma_{hh}$ if and only if:
\begin{equation}
\sum_{k\in\I\backslash\{i\}}
\frac{p_{k}\rho_{k}}{\sigma^2}(g_{ih}g_{kj}-g_{ij}g_{kh}) 
\leq (g_{ij}-g_{ih}).
\label{eq:decoding}
\end{equation}
Recall that $g_{ij}\geq g_{ih}$, and therefore the right-hand side of~\eqref{eq:decoding} is non-negative. \textcolor[rgb]{0,0,0}{Because of the condition $g_{ij}/g_{ih}\geq g_{kj}/g_{kh}$ for all $k\in\C\backslash\{i\}$ in the statement of the lemma, the left-hand side is non-positive.} Hence Lemma~\ref{rmk:decoding}.
\end{proof}

\end{tcolorbox}
\begin{tcolorbox}[boxrule=0.5pt, colback=yellow!5]
The result of Lemma~\ref{rmk:decoding} is coherent with the previous observations that
two UEs with large difference in channel conditions are preferred
to be paired~\cite{7273963,7357604}. 

If $g_{ij}\gg g_{ih}$, then most likely the condition in Lemma~\ref{rmk:decoding} holds, as the large scale path-loss from other cells, tends not to differ as much as from the serving cell $i$ in this case. 
Besides, the large scale path-loss is a practically reasonable factor for ranking the decoding order~\cite{7307220,6954404}.
\textcolor[rgb]{0,0,0}{In Section~\ref{sec:simulation}, numerical results further show that considering the UE pairs as defined by Lemma~\ref{rmk:decoding} virtually does not lead to any loss of performance.
 Lemma 1 is used to filter the candidate pairs set $\U$ (i.e. to drop some candidate pairs from $\U$) so as to reduce computational complexity.} From now on, we let $\U_i$ be composed of pairs satisfying Lemma~\ref{rmk:decoding}.
\end{tcolorbox}

\subsection{RB Allocation}
\label{subsec:RB-alloc}

\textcolor[rgb]{0,0,0}{If UE $j$ ($j\in\J$) is using orthogonal RB allocation, then the achievable capacity}\footnote{For the sake of presentation, we use the natural logarithm throughout the paper. We remark that all conclusions hold for the logarithm to base $2$.} \textcolor[rgb]{0,0,0}{of $j$ is~\eqref{eq:cj}, with $\gamma_j$ being~\eqref{eq:sinrOMA}.}
\begin{equation}
c_{j}=\log(1+\gamma_j).
\label{eq:cj}
\end{equation}

For non-orthogonal RB allocation, the achievable capacity for $j$ and $\u$ ($j\in\u$) is computed by $c_{j\u}=MB\log\left(1+\gamma_{j\u}\right)$ with $\gamma_{j\u}$ being~\eqref{eq:sinr+} or~\eqref{eq:sinr-}. Therefore,
\begin{equation}
    c_{j\u}= \left\{
\begin{array}{ll}
 \log\left(1+\gamma_{j\u}\right) & j\in\u \\
 0 & j\notin\u
\end{array}\right..
\label{eq:cju}
\end{equation}

For UE $j$ ($j\in\J$), we use $x_j$ to denote the proportion of RBs with orthogonal RB allocation to $j$. For any pair $\u$ ($\u\in\U$), denote by $x_{\u}$ the non-orthogonal RB allocation for the two UEs in pair $\u$. We use the vector $\vec{x}$ to represent the RB allocation for all the UEs, i.e., 
\[
\vec{x}=[\underbrace{x_1,x_{2},\ldots,x_{m}}_{\text{Orthogonal RB allocation}},\underbrace{x_{\u_1},x_{\u_2},\ldots,x_{\u_s}}_{\text{Non-orthogonal RB allocation}}].
\] 
\textcolor[rgb]{0,0,0}{For any UE $j$, $x_j=0$ means that UE $j$ does not use orthogonal RB allocation.} Similarly, for any pair $\u$, $x_{\u}=0$ means that pair $\u$ is not put in use. For any UE $j$, if $x_{\u}=0$ for all $\u\in\V_j$, then it means that UE $j$ only uses orthogonal RB allocation. \textcolor[rgb]{0,0,0}{Resources used by different pairs are orthogonal such that there is no interference among pairs.}
\textcolor[rgb]{0,0,0}{Denote by $\bar{\rho}$ the cell load limit. By constraining that the sum of them which equals to the load of cell $i$ does not exceed $\bar{\rho}$, the amounts represented by $x_j$ ($j\in\J_i$) and $x_{\u}$ ($\u\in\U_i$) do not overlap. Orthogonal RB allocation is considered among the pairs in one cell, meaning that the pairs do not have interference with each other.}
\begin{equation}
    \underarrow[\rho_i][\uparrow]{$\substack{\text{Cell} \\ \text{load}}$}=\underbrace{\sum_{j\in\J_i} x_j}_{\substack{\text{Orthogonal} \\ \text{RB proportion}}}
     + \underbrace{\sum_{\u\in\U_i} x_{\u}}_{\substack{\text{Non-orthogonal} \\ \text{RB proportion}}} \leq \underarrow[\bar{\rho}][\uparrow]{$\substack{\text{Load} \\ \text{limit}}$}.
\label{eq:rhoi}
\end{equation}
We use $\bm{\rho}$ to represent the vector of network load, i.e.,
\[
\bm{\rho}=[\rho_1,\rho_2,\ldots,\rho_n].
\]
\textcolor[rgb]{0,0,0}{Similarly, we use vector $\bar{\bm{\rho}}$ to denote the load limits of all cells.}

The term $c_{j}x_j$ computes the bits delivered to UE $j$ with orthogonal RB allocation, because $c_j$ is the achievable capacity of UE $j$ on all RBs and $x_j$ is the proportion of RBs with orthogonal RB allocation.
 Similarly, the term $c_{j\u}x_\u$ is the bits delivered to UE $j$ by non-orthogonal RB allocation for pair $\u$. Denote by $d_j$ the bits demand of UE $j$. The quality-of-service (QoS) requirement is:
\begin{equation}
\underbrace{c_{j}x_j}_{\substack{\text{Bits delivered} \\ \text{by orthogonal} \\ \text{RB allocation}}}+\underbrace{\sum_{\u\in\V_j}c_{j\u}x_{\u}}_{\substack{\text{Bits delivered by} \\ \text{non-orthogonal} \\ \text{RB allocation} }}\geq \underarrow[d_j][\uparrow]{$\substack{\text{Bits} \\ \text{demand}}$}.
\label{eq:dj}
\end{equation}
\textcolor[rgb]{0,0,0}{We remark that $d_j$ is normalized by the RB spectral bandwidth and the total number of RBs, for the sake of presentation.} \textcolor[rgb]{0,0,0}{Note that a user can use orthogonal RB allocation individually, or non-orthogonal RB allocation with the other user in the pair, or both, which is subject to optimization. The amount of allocated RBs to a user in OMA or a pair adopting NOMA, is subject to optimization, under the constraint that the overall allocated resource does not exceed limit.}


We remark that there is an implicit \textit{pair selection} problem in the above expressions. Note that $|\U|$ increases fast with $|\J|$. It is therefore impractical to simultaneously use all pairs in $\U$. To deal with this issue, each UE is allowed to use up to one pair in $\U$ for optimization, as formulated later in Section~\ref{sec:formulation}, though our system model is not limited by this. The problem of pairing and resource allocation is challenging: First, UEs of the same pair are coupled in resource allocation. Second, one can observe that increasing $x_{\u}$ (or $x_j$) for some pair $\u$ (or some UE $j$) may enhance the throughput of the UEs of $\u$ (or UE $j$). However, since $x_{\u}$ (or $x_j$) appears in the inter-cell interference term (see \eqref{eq:Ikj} and \eqref{eq:rhoi}), the increase of $x_{\u}$ (or $x_j$) results in less available resources for other UEs and leads to more interference. \textcolor[rgb]{0,0,0}{The user pairing selection is not given a priori but is determined by optimization.} \textcolor[rgb]{0,0,0}{We remark that whether or not a UE should be allocated with resources with OMA or NOMA, or both, is up to optimization. The overall amount of resource used by NOMA and OMA in the entire network are part of the optimization output.}


\subsection{Comparison to OMA Modeling}
The models proposed for OMA in~\cite{7959870,7132788,5198628,6204009,6732895,7585124,7880696} are inherently a special case of our NOMA model. The former is obtained by setting $\U=\phi$. Then, the terms for non-orthogonal RB allocation disappear in~\eqref{eq:rhoi} and~\eqref{eq:dj} and $\vec{x}$ is therefore eliminated in \eqref{eq:cj}--\eqref{eq:dj}. Also, there is no power split in OMA. Hence \eqref{eq:cj}--\eqref{eq:dj} form a non-linear system only in terms of $\bm{\rho}$. This system falls into the analytical framework of standard interference function (SIF)~\cite{414651}, which enables the computation of the optimal network load settings via fixed-point iterations~\cite{6204009}. However, for the general NOMA case, the resource allocation is not at UE-level.
One needs to split a UE's demand between orthogonal and non-orthogonal RB allocations, which results in a new dimension of complexity.


\section{Problem Formulation}
\label{sec:formulation}

 By successively plugging \eqref{eq:sinrOMA} and \eqref{eq:Ikj} into \eqref{eq:cj}, we obtain a function $c_{j}$ in load $\bm{\rho}$, i.e., $c_{j}(\bm{\rho})$. Similarly, we obtain $c_{j\u}(\vec{q},\bm{\rho})$ from \eqref{eq:sinr+}, \eqref{eq:sinr-}, \eqref{eq:Ikj}, and \eqref{eq:cju}. For pair $\u$ ($\u\in\U$), we use a binary variable $y_{\u}$ to indicate whether or not the pair $\u$ is selected. Define $\vec{y}$ as
\[
\vec{y} = [y_{\u_1},y_{\u_2},\ldots,y_{\u_s}].
\]
We minimize a generic cost function \textit{$F(\bm{\rho})$ that is monotonically (but not necessarily strictly monotonically) increasing in each element of $\bm{\rho}$}. \minf~is given below.
Constraints~\eqref{eq:minf-rhobar} guarantee that the cell load complies to the load limit $\bar{\rho}$. Constraints~\eqref{eq:minf-rho} state the relationship between RB allocation and cell load. Constraints~\eqref{eq:minf-d} and \eqref{eq:minf-p} are for QoS and power, respectively. Constraints~\eqref{eq:minf-xy} guarantee that RB allocation occurs only for selected pairs. By constraints~\eqref{eq:minf-y}, each UE belongs up to one pair such that the selected pairs are mutually exclusive. The variables are cell load $\bm{\rho}$, power allocation $\vec{q}$, RB allocation $\vec{x}$, and pair selection $\vec{y}$. The variable domains are imposed by~\eqref{eq:minf-rhoqx} and~\eqref{eq:minf-integer}. Throughout this paper, we use $\vec{0}$ to represent zero vector/matrix. For simplicity, the dimension(s) of $\vec{0}$ is not explicitly stated. 

\begin{subequations}
\begin{alignat}{2}
[\minf]\quad &
\min\limits_{\bm{\rho},\vec{q},\vec{x},\vec{y}} F(\bm{\rho}) \\
\textnormal{s.t.} \quad 
&  \rho_i\leq \bar{\rho},~i\in\I \label{eq:minf-rhobar}\\
& \rho_i=\sum_{j\in\J_i} x_j + \sum_{\u\in\U_i} x_{\u},~i\in\I \label{eq:minf-rho}\\
&     c_{j}(\bm{\rho})x_j+\sum_{\u\in\V_j}c_{j\u}(\vec{q},\bm{\rho})x_{\u}\geq d_j,~j\in\J \label{eq:minf-d}\\
& \sum_{j\in\u}q_{j\u}= p_i,~\u\in\U_i,~i\in\I \label{eq:minf-p}\\
& x_{\u}\leq y_{\u},~\u\in\U \label{eq:minf-xy}\\
& \sum_{\u\in\V_j}y_{\u}\leq 1,~j\in\J \label{eq:minf-y}\\
& \bm{\rho},\vec{q},\vec{x}\geq\vec{0} \label{eq:minf-rhoqx}\\
& y_{\u} \in \{0,1\},~\u\in\U \label{eq:minf-integer}
\end{alignat}
\label{eq:minF}
\end{subequations}

\section{Optimization within a Cell}
\label{sec:single}

\textcolor[rgb]{0,0,0}{In multi-cell NOMA, due to the interference among cells, one cell's pair selection may affect the other cells' power splits, and vice versa. Let us consider a simple case in this section. 
Suppose we optimize the load of one cell $i$, and the cell load levels of $\I\backslash\{i\}$ are temporarily fixed.
This optimization step is a module for solving \minf~later in Section~\ref{sec:multiple}.}
We respectively use $\vec{q}_i$, $\vec{x}_i$, $\vec{y}_i$ to denote the corresponding variable elements for power allocation, RB allocation, and pair selection. Vector $\bm{\rho}_{-i}$ is composed of all elements but $\rho_i$ of $\bm{\rho}$. We minimize $\rho_i$ under fixed $\bm{\rho}_{-i}$, as formulated below. 
\begin{equation}
\min\limits_{\rho_i,\vec{q}_i,\vec{x}_i,\vec{y}_i} \rho_i~\text{s.t. \eqref{eq:minf-rho}--\eqref{eq:minf-integer} of cell $i$, with fixed $\bm{\rho}_{-i}$.} 
\label{eq:minload}
\end{equation}
Since  $\bm{\rho}_{-i}$ is fixed, $c_j$ is a constant
and $c_{j\u}$ is a function in $\vec{q}_i$ only. \textcolor[rgb]{0,0,0}{Different from previous works \cite{7560605} and \cite{7587811}, this single-cell resource optimization problem is subject to user demand constraints.}

\textcolor[rgb]{0,0,0}{The optimization is not straightforward even under fixed inter-cell interference. The optimal power split for one pair is up to how much time-frequency resource is allocated to this pair. In other words, for one pair $\u$, if the amount of RBs allocated to $\u$ changes, the optimal power split for $\u$ before this change loses its optimality. So the power split $\vec{q}$ and the resource allocation $\vec{x}$ are coupled together. In addition, the pair selection is a combinatorial problem. Therefore, the power split $\vec{q}$, the time-frequency resource allocation $\vec{x}$, and the user pair selection $\bm{y}$, must be optimized jointly.}

\begin{lemma}
All constraints of \eqref{eq:minf-d} in \eqref{eq:minload} hold as equalities at any optimum.
\label{lma:equality}
\end{lemma}
\begin{proof}
Denote the optimal objective value of~\eqref{eq:minload} by $\rho'_i$ and the optimal orthogonal RB allocation of $j$ ($j\in\J_i$) by $x'_j$. Suppose strict inequality holds for some $j$. If $x'_j>0$, by fixing all other variables except for $x_j$ in~\eqref{eq:minload}, one can verify that the solution $x'_j-\epsilon$ ($\epsilon>0$) is still feasible to~\eqref{eq:minload} as long as $\epsilon$ is sufficiently small. In addition,  $x'_j-\epsilon$ leads to a lower objective value $\rho'_i-\epsilon$, which contradicts our assumption that $\rho'_i$ is the optimal objective value. If $x'_j=0$, then $j$'s demand has to be satisfied by non-orthogonal RB allocation and the same argument applies to variable $x_{\u}$ ($j\in\u$). 
\end{proof}
The first analytical result is that, \textit{the optimal power split is independent of pair selection}, in~Theorem~\ref{thm:independency}. Denote by $\Y_{\u}$ the set of all possible pairing solutions of~\eqref{eq:minload} that includes pair $\u$, i.e., 
\[
\Y_{\u}=\{\vec{y}_i | y_{\u}=1, \sum_{\u\in\V_j}y_{\u}\leq 1,~j\in\J_i \}.
\]
\begin{definition}
Given a pair selection $\hat{\vec{y}}_i$ ($\hat{\vec{y}}_i\in\Y_\u$), 
the optimal power split for pair $\u$ ($\u\in\U_i$), denoted by $\hat{\vec{q}}_{\u}$, is the column for pair $\u$ in $\hat{\vec{q}}_i$, where $\hat{\vec{q}}_i$ is obtained by optimally solving \eqref{eq:minload} for $\hat{\vec{y}}_i$.
\label{def:opt_splitting}
\end{definition} 
\begin{theorem}
Consider $\u$ ($\u\in\U_i$).
The optimal power split for any $\hat{\vec{y}}_i$ ($\hat{\vec{y}}_i\in\Y_{\u}$) is also optimal for $\hat{\vec{y}}'_i$ ($\hat{\vec{y}}'_i\in\Y_{\u}$).
\label{thm:independency}
\end{theorem}
\begin{proof}
Denote by $\hat{\vec{q}}_{\u}$ and $\hat{\vec{q}}'_{\u}$ the optimal power splits for $\hat{\vec{y}}_i$ and $\hat{\vec{y}}'_i$, respectively. Suppose $\hat{\vec{q}}_{\u}$ is not optimal for $\hat{\vec{y}}'_i$. 
There are two possibilities: 1) $c_{j\u}(\hat{\vec{q}}_{\u})=c_{j\u}(\hat{\vec{q}}'_{\u})$ ($j\in\u$); 2) $c_{j\u}(\hat{\vec{q}}_{\u})\neq c_{j\u}(\hat{\vec{q}}'_{\u})$ for at least one $j$ in $\u$. 

For 1), $\hat{\vec{q}}_{\u}$ and $\hat{\vec{q}}'_{\u}$ result in the same $x_j$ and $x_{\u}$ for satisfying~\eqref{eq:minf-d} and are equally good for~\eqref{eq:minload}, which conflicts our assumption. Thus $\hat{\vec{q}}_{\u}$ is optimal for $\hat{\vec{y}}'_i$. We then consider 2) and assume $c_{j\u}(\hat{\vec{q}}_{\u})>c_{j\u}(\hat{\vec{q}}'_{\u})$. By Lemma~\ref{lma:equality}, $\hat{\vec{q}}'_{\u}$ makes \eqref{eq:minf-d} become equality under $\hat{\vec{y}}'_i$. Replacing $\hat{\vec{q}}'_{\u}$ by $\hat{\vec{q}}_{\u}$ leads to some slack in \eqref{eq:minf-d} and hence the objective can be improved. This contradicts that $\hat{\vec{q}}'_{\u}$ is optimal for $\vec{y}'_i$. The same proof applies to $c_{j\u}(\hat{\vec{q}}_{\u})<c_{j\u}(\hat{\vec{q}}'_{\u})$. Hence the conclusion.
\end{proof}

By Theorem~\ref{thm:independency}, the optimal power split is decoupled from pair selection. 
Next we analytically prove how to find the optimal power split for any pair. 

\subsection{Finding Optimal Power Split}
\label{subsec:pair}

Under fixed $\vec{y}_i$ ($\vec{y}_i\in\Y_{\u}$, $\u\in\U_i$), constraints~\eqref{eq:minf-y} are removed. Constraints~\eqref{eq:minf-xy}, and \eqref{eq:minf-integer} of \eqref{eq:minload} for all $\u$ with $y_{\u}=0$ in $\vec{y}_i$ are removed. \textcolor[rgb]{0,0,0}{Therefore, for each pair $\u=\{\oplus,\ominus\}$, we can formulate a problem in~\eqref{eq:min_pair}. Solving this problem yields the optimal power split.} In~\eqref{eq:min_pair}, $x_{\oplus}$ and $x_{\ominus}$ are the orthogonal RB allocation for $\oplus$ and $\ominus$, respectively. The variable $x_{\u}$ denotes the amount of non-orthogonal RB allocation for $\u$.
\begin{subequations}
\begin{alignat}{2}
&
\min\limits_{\substack{x_{\oplus},x_{\ominus},x_{\u}\geq 0 \\ \vec{q}_{\u}\geq\vec{0}}}  x_{\oplus} + x_{\ominus} + x_{\u} \\
\text{s.t.} \quad &  c_{\oplus}(\bm{\rho}_{-i})x_{\oplus}+ c_{\oplus\u}(\vec{q}_{\u},\bm{\rho}_{-i})x_{\u} \geq d_{\oplus} \label{eq:min_decompd1}\\
& c_{\ominus}(\bm{\rho}_{-i})x_{\ominus} + c_{\ominus\u}(\vec{q}_{\u},\bm{\rho}_{-i})x_{\u} \geq d_{\ominus} \\
& q_{\oplus\u}+q_{\ominus\u}= p_i \label{eq:min_pair-p}
\end{alignat}
\label{eq:min_pair}
\end{subequations}

For deriving solution method for \eqref{eq:min_pair}, define function $w_j$ ($j=\oplus$ or $j=\ominus$) of $\bm{\rho}_{-i}$ as follows.
\begin{equation}
w_{j}(\bm{\rho}_{-i})=\left(\sum_{k\in\I\backslash\{i\}}p_{k}g_{kj}\rho_k+\sigma^2\middle)\right/g_{ij}.
\end{equation}
For $q_{\oplus\u}$, one can derive from~\eqref{eq:sinr+} and \eqref{eq:cju}: 
\begin{equation}
q_{\oplus\u}=(e^{c_{\oplus\u}}-1)w_{\oplus}(\bm{\rho}_{-i}).
\label{eq:convert}
\end{equation}
Combining~\eqref{eq:convert} with~\eqref{eq:min_pair-p}, $q_{\oplus\u}$ and $q_{\ominus\u}$ can be eliminated, giving~\eqref{eq:min_pair2} below. Formulation~\eqref{eq:min_pair2} is equivalent to~\eqref{eq:min_pair}. Given $c_{\oplus\u}$ and $c_{\ominus\u}$, the corresponding $q_{\oplus\u}$ and $q_{\ominus\u}$ can be obtained from $c_{\oplus\u}$ and $c_{\ominus\u}$ by \eqref{eq:min_pair-p} and  \eqref{eq:convert}.  
\begin{subequations}
\begin{alignat}{2}
&
\min\limits_{\substack{x_{\oplus},x_{\ominus},x_{\u}\geq 0\\ c_{\oplus\u},c_{\ominus\u}\geq 0}}  x_{\oplus} + x_{\ominus}+x_{\u} \\
\text{s.t.} \quad &  c_{\oplus}(\bm{\rho}_{-i})x_{\oplus}+ c_{\oplus\u}x_{\u} \geq d_{\oplus} \label{eq:min_pair2-d+}\\
& c_{\ominus}(\bm{\rho}_{-i})x_{\ominus} + c_{\ominus\u}x_{\u} \geq d_{\ominus} \label{eq:min_pair2-d-} \\
& \text{Cv$_{\u}$}(c_{\oplus\u},c_{\ominus\u},\bm{\rho}_{-i})\leq 0 \label{eq:min_pair2-geq0}
\end{alignat}
\label{eq:min_pair2}
\end{subequations}

In~\eqref{eq:min_pair2}, the function $\text{Cv$_{\u}$}$ is defined in~\eqref{eq:conv} in the Appendix. One can easily verify that $\text{Cv$_{\u}$}$ is convex in $c_{\oplus\u}$ and $c_{\ominus\u}$ (with  $g_{i\oplus}\geq g_{i\ominus}$). The difficulty of~\eqref{eq:min_pair2} is on the two bi-linear constraints~\eqref{eq:min_pair2-d+} and~\eqref{eq:min_pair2-d-}. However, they become linear with fixed $x_{\u}$. To ease the presentation, we define the function below.
\begin{equation}
Z_{\u}(x_{\u},\bm{\rho}_{-i})=x_{\u}+\min_{\substack{x_{\oplus},x_{\ominus}\geq0 \\c_{\oplus\u},c_{\ominus\u}\geq 0}}x_{\oplus}+x_{\ominus}\text{ s.t. \eqref{eq:min_pair2-d+}--\eqref{eq:min_pair2-geq0}}.
\label{eq:Z}
\end{equation}
Solving~\eqref{eq:min_pair2} (and equivalently~\eqref{eq:min_pair}) is to find the minimum of $Z_{\u}(x_{\u},\bm{\rho}_{-i})$.
The following theorem shows the uniqueness of the minimum of $Z_{\u}(x_{\u},\bm{\rho}_{-i})$. 
\begin{theorem}
\textcolor[rgb]{0,0,0}{$Z_{\u}(x_{\u},\bm{\rho}_{-i})$ has unique minimum in $x_{\u}$.}
\label{thm:Z}
\end{theorem}
\begin{proof}
\textcolor[rgb]{0,0,0}{Since the first term $x_{\u}$ in $Z_{\u}(x_{\u},\bm{\rho}_{-i})$ is strictly monotonically increasing, to prove that $Z_{\u}(x_{\u},\bm{\rho}_{-i})$ has unique minimum, we only need to prove that the remaining part of $Z_{\u}(x_{\u},\bm{\rho}_{-i})$ is monotonically (but not necessarily strictly monotonically) decreasing in $x_{\u}$.} For this part, at the optimum~\eqref{eq:min_pair2-d+} and~\eqref{eq:min_pair2-d-} hold as equalities because of Lemma~\ref{lma:equality}. Hence, reformulating the problem by replacing the inequalities in~\eqref{eq:min_pair2-d+} and~\eqref{eq:min_pair2-d-} with equalities does not lose optimality. With equalities, the variables $x_{\oplus}$ and $x_{\ominus}$ can be represented by $c_{\oplus\u}$ and $c_{\ominus\u}$:
\begin{equation}
x_{\oplus}=\frac{(d_{\oplus}-c_{\oplus\u}x_{\u})}{c_{\oplus}(\bm{\rho}_{-i})},~x_{\ominus}=\frac{(d_{\ominus}-c_{\ominus\u}x_{\u})}{c_{\ominus}(\bm{\rho}_{-i})}.
\label{eq:re-obj}
\end{equation}
Therefore $x_{\oplus}$ and $x_{\ominus}$ can be eliminated from the objective function.
The minimization is thus equivalent to maximizing $c_{\oplus\u}/c_{\oplus}(\bm{\rho}_{-i})+c_{\ominus\u}/c_{\ominus}(\bm{\rho}_{-i})$.
We formulate this maximization problem below.
\begin{subequations}
\begin{alignat}{2}
&
\max\limits_{\substack{c_{\oplus\u},c_{\ominus\u}\geq 0}} 
\frac{c_{\oplus\u}}{c_{\oplus}(\bm{\rho}_{-i})}+\frac{c_{\ominus\u}}{c_{\ominus}(\bm{\rho}_{-i})}\label{eq:min_pair3-obj}\\
\text{s.t.} \quad &  c_{\oplus\u}x_{\u} \leq d_{\oplus} \label{eq:min_pair3-d+}\\
& c_{\ominus\u}x_{\u} \leq d_{\ominus} \label{eq:min_pair3-d-} \\
& \text{Cv$_{\u}$}(c_{\oplus\u},c_{\ominus\u},\bm{\rho}_{-i})\leq 0 \label{eq:min_pair3-geq0}
\end{alignat}
\label{eq:min_pair3}
\end{subequations}

Constraints~\eqref{eq:min_pair3-d+} and~\eqref{eq:min_pair3-d-} originate from the non-negativity requirement of $x_{\oplus}$ and $x_{\ominus}$. 
Note that \eqref{eq:min_pair3} is convex. In addition, the feasible region shrinks with the increase of $x_{\u}$.  Then the optimum of~\eqref{eq:min_pair3} monotonically decreases with $x_{\u}$. Hence the theorem.
\end{proof}


Because of Theorem~\ref{thm:Z}, bi-section search of $x_{\u}$ reaches the minimum of $Z_{\u}(x_{\u},\bm{\rho}_{-i})$.
Note that for any $x_{\u}$, computing $Z_{\u}(x_{\u},\bm{\rho}_{-i})$ needs to solve~\eqref{eq:min_pair3}. In the following, we prove how this can be done much more efficiently than employing standard convex optimization. 

We remark that~\eqref{eq:min_pair3} is a two-dimensional optimization problem with respect to $c_{\oplus\u}$ and $c_{\ominus\u}$. Constraints~\eqref{eq:min_pair3-d+} and~\eqref{eq:min_pair3-d-} are defined by two hyperplanes $c_{\oplus\u}=d_{\oplus}/x_{\u}$ and $c_{\ominus\u}=d_{\ominus}/x_{\u}$, respectively. Due to the convexity of the function $\text{Cv$_{\u}$}$ in $c_{\oplus\u}$ and $c_{\ominus\u}$, the curve $\text{Cv$_{\u}$}(c_{\oplus\u},c_{\ominus\u},\bm{\rho}_{-i})=0$ in~\eqref{eq:min_pair3-geq0} along with $c_{\oplus\u}=0$ and $c_{\ominus\u}=0$ forms a convex set of $[c_{\oplus\u},c_{\ominus\u}]$. The optimum of~\eqref{eq:min_pair3} depends on whether the two hyperplanes intersect with the curve and how they intersect. This leads to three possible cases to be considered, named Case 1, Case 2, and Case 3, respectively. In Case 1, constraints~\eqref{eq:min_pair3-d+} and~\eqref{eq:min_pair3-d-} are redundant, and the optimum is determined by the coefficients ${1}/{c_{\oplus}(\bm{\rho}_{-i})}$ and ${1}/{c_{\ominus}(\bm{\rho}_{-i})}$ in the objective function and the curve $\text{Cv$_{\u}$}(c_{\oplus\u},c_{\ominus\u},\bm{\rho}_{-i})=0$. In Case 2, the optimum is defined by one of the hyperplanes and the curve. In Case 3, constraint~\eqref{eq:min_pair3-geq0} is redundant, and the optimum is determined by the two hyperplanes. With $x_{\u}$ increasing from $0$ to $\infty$, Case 1, Case 2, and Case 3 happen sequentially, and all happen eventually. The three cases are illustrated in~\figurename~\ref{fig:region} with colors. 
\begin{figure}[!t]
\centering
\includegraphics[width=\linewidth]{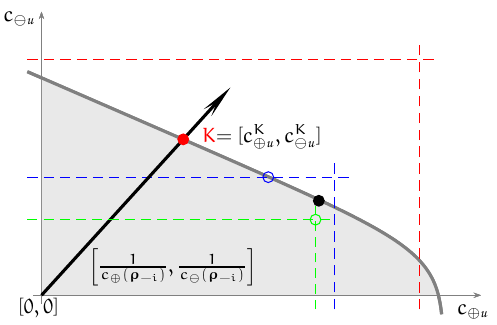}
\caption{This figure shows the feasible region of~\eqref{eq:min_pair3} of $c_{\oplus\u}$ and $c_{\ominus\u}$ for different $x_{\u}$. The shadowed area below the curve is~\eqref{eq:min_pair3-geq0}. The vertical and horizontal dashed lines are the hyperplanes defined by~\eqref{eq:min_pair3-d+} and~\eqref{eq:min_pair3-d-}, respectively. The hyperplanes are shown by red, blue, and green dashed lines for Case 1, Case 2, and Case 3, respectively. The point $K$ is optimal for Case 1. The blue and green circles represent the optimal solutions for Case 2 and Case 3, respectively. The black dot is the point on the curve where the two hyperplanes intersect with each other (see footnote~\ref{footnote:black-dot} for the existence of this point.).}
\label{fig:region}
\end{figure}
Below we respectively show how to compute the optimum for each case. 

In the following, we first compute the optimum in Case 1, represented by $K=[c^K_{\oplus\u},c^K_{\ominus\u}]$, which is the intersection of the vector $[{1}/{c_{\oplus}(\bm{\rho}_{-i})},{1}/{c_{\ominus}(\bm{\rho}_{-i})}]$ and the curve. The point also leads to a closed-form solution for the optima of all the three cases. 
Mathematically, point $K$ is solved by applying bi-section search to~\eqref{eq:pointK} below. 
\begin{equation}
\left\{
\begin{array}{l}
\text{Cv$_{\u}$}(c_{\oplus\u},c_{\ominus\u},\bm{\rho}_{-i})=0 \\
{c_{\oplus\u}}{c_{\oplus}(\bm{\rho}_{-i})}={c_{\ominus\u}}{c_{\ominus}(\bm{\rho}_{-i})}.
\end{array}
\right.	
\label{eq:pointK}
\end{equation}
For solving \eqref{eq:pointK}, we can first eliminate $c_{\oplus\u}$ or $c_{\ominus\u}$ in the first equation. This can be done by representing one of $c_{\oplus\u}$ and $c_{\ominus\u}$ with the other by the second equation. Since there is only one variable in the first equation, one can use bi-section search to find its solution\footnote{The solution is guaranteed to be unique and hence bi-section search applies. This is because, by representing one of $c_{\oplus\u}$ and $c_{\ominus\u}$ by the other by function \text{Cv$_{\u}$}, one variable is monotonically decreasing in the other, resulting in a unique zero point.}. 
Then, we compute the value of $x_{\u}$ when at least one hyperplane goes through $K$, denoted by $x^{K}_{\u}$ in~\eqref{eq:xK}. 
\begin{equation}
x^K_{\u}=\min\{d_{\oplus}/c_{\oplus\u}^K, d_{\ominus}/c_{\ominus\u}^K\}.
\label{eq:xK}
\end{equation}
The three cases, indicated by colors in \figurename~\ref{fig:region}, are as follows. 

\textbf{\textcolor[rgb]{1,0,0}{Case 1}} ($x_{\u}\leq \min_{j\in\u}d_{j}/c_{j\u}^K$): Point $K$ is the optimum of \eqref{eq:min_pair3}, because~\eqref{eq:min_pair3-d+} and~\eqref{eq:min_pair3-d-} are redundant. This happens when $x_{\u}$ is sufficiently small (or 0), as shown in~\figurename~\ref{fig:region}.

\textbf{\textcolor[rgb]{0,0,1}{Case 2}} ($x_{\u}>\min_{j\in\u}d_{j}/c_{j\u}^K$ and $\text{Cv}_{\u}(d_{\oplus}/x_{\u},d_{\ominus}/x_{\u})>0$): There exists one point on the curve where both two hyperplanes intersect\footnote{\label{footnote:black-dot}The existence of this point is guaranteed: With the increase of $x_{\u}$, both hyperplanes will eventually intersect with the curve with two intersection points. By increasing $x_{\u}$, the distance between the two intersections keeps being smaller. The two intersections will eventually overlap.}. We represent this point by the black dot on the curve in~\figurename~\ref{fig:region}. In Case 2, one hyperplane intersects with the curve at some point between $K$ and the black dot, and intersects with the other hyperplane on some point above the curve, see \figurename~\ref{fig:region}.
The intersection point of the curve and the hyperplane is the optimum of~\eqref{eq:min_pair3}. Without loss of generality, we assume $K$ violates \eqref{eq:min_pair3-d-}, meaning that the hyperplane of \eqref{eq:min_pair3-d-} goes through the optimum, as shown by  \figurename~\ref{fig:region}. By plugging the equation of the hyperplane into that of the curve, the optimal $c_{\oplus\u}$ is a function of $x_{\u}$. Similarly, if $K$ violates \eqref{eq:min_pair3-d-} instead, then $c_{\oplus\u}=d_{\oplus}/x_{\u}$ and the optimal $c_{\ominus\u}$ is a function of $x_{\u}$. To know which hyperplane goes through the optimum, one only needs to check which of $c^{K}_{j\u}x_{j\u}^{K}> d_{j}$ ($j=\oplus$ or $j=\ominus$) holds. Note that exactly one of the two holds in Case 2. The optimal $c_{\oplus\u}$ and $c_{\ominus\u}$ are computed respectively by~\eqref{eq:H+} and~\eqref{eq:H-} defined in the Appendix.

\textbf{\textcolor[rgb]{0,1,0}{Case 3}} ($x_{\u}>\min_{j\in\u}d_{j}/c_{j\u}^K$ and $\text{Cv}_{\u}(d_{\oplus}/x_{\u},d_{\ominus}/x_{\u})\leq 0$): Constraint~\eqref{eq:min_pair3-geq0} is redundant, as shown in \figurename~\ref{fig:region}. The optimum is the intersection point of the two hyperplanes, computed by $c_{j\u}=d_j/x_{\u}$ ($j=\oplus$ or $j=\ominus$). 

In summary, the optimal solution of \eqref{eq:min_pair3} is computed by \eqref{eq:C} below ($j=\oplus$ or $j=\ominus$) in closed form, with $H_{j\u}$ being~\eqref{eq:H+} or~\eqref{eq:H-} in the Appendix.
\begin{equation}
C_{j\u}(x_{\u},\bm{\rho}_{-i})=\left\{
\begin{array}{ll}
c^{K}_{j\u} &  \text{Case 1} \\
H_{j\u}(x_{\u},\bm{\rho}_{-i}) & \text{Case 2}  \\
d_{j}/x_{\u} & \text{Case 3}
\end{array}
\right.	
\label{eq:C}
\end{equation}

The function $Z_{\u}(x_{\u},\bm{\rho}_{-i})$ computes the amount of resource used for both orthogonal and non-orthogonal RB allocations for the UEs in $\u$.
It is optimal to serve the two UEs only by orthogonal RB allocation, if the minimum of $Z_{\u}(x_{\u},\bm{\rho}_{-i})$ occurs at $x_{\u}=0$.
In all other cases, $\min_{x_{\u}}Z_{\u}(x_{\u},\bm{\rho}_{-i})$ yields the optimal power split for non-orthogonal RB allocations.
The algorithm optimally solving \eqref{eq:min_pair}, named \textsc{Split}, is as follows.
\begin{tcolorbox}[boxrule=0.5pt]
\vspace{-0.2cm}
\begin{codebox}
\Procname{$\proc{Split}(\u,\bm{\rho}_{-i})$}
\li $x^{*}_{\u}\gets \argmin_{x_{\u}} Z_{\u}(x_{\u},\bm{\rho}_{-i})$  \textcolor[rgb]{0.18,0.53,0.33}{\Comment\small Bi-section search}
\li Compute $\langle c^{*}_{\oplus\u},c^{*}_{\ominus\u}\rangle$ by~\eqref{eq:C} \textcolor[rgb]{0.18,0.53,0.33}{\Comment\small With $x^{*}_{\u}$, $\bm{\rho}_{-i}$}
\li Compute $\langle x^{*}_{\oplus}$, $x^{*}_{\ominus}\rangle$ by~\eqref{eq:re-obj} \textcolor[rgb]{0.18,0.53,0.33}{\Comment\small With $c^{*}_{\oplus\u}$, $c^{*}_{\ominus\u}$, $\bm{\rho}_{-i}$}
\li Convert $\langle c^{*}_{\oplus\u},c^{*}_{\ominus\u}\rangle$ to $\langle q^{*}_{\oplus\u},q^{*}_{\ominus\u}\rangle$ \textcolor[rgb]{0.18,0.53,0.33}{\Comment\small By \eqref{eq:min_pair-p},~\eqref{eq:convert} }
\li \Return $\langle q^{*}_{\oplus\u},q^{*}_{\ominus\u},x^{*}_{\oplus},x^{*}_{\ominus},x^{*}_{\u} \rangle$
\end{codebox}
\end{tcolorbox}

\subsection{Optimal Pairing}
\label{subsec:cell}

Denote by $\Y_i$ the set of all candidate pair selections:  
\[
\Y_i=(\cup_{\u\in\U_i}\Y_{\u})\cup\{\vec{0}\}.
\]

By obtaining $\min_{x_\u}Z_{\u}(x_{\u},\bm{\rho}_{-i})$ for all $\u\in\U_i$ as shown earlier in Section~\ref{subsec:pair}, enumerating all $\vec{y}_i$ in $\Y_i$ gives the optimal solution to \eqref{eq:minload}. This exhaustive search however does not scale, as $|\Y_i|$ is exponential in the number of UEs. By the following derivation, we are able to obtain the optimum of~\eqref{eq:minload} in polynomial time. 
\begin{theorem}
The optimum of \eqref{eq:minload} is computed by finding the maximum weighted matching in an undirected graph.
\label{thm:matching}
\end{theorem}
\begin{proof}
To prove the conclusion, 
an undirected weighted graph $\G_i$ is constructed and explained below.
\begin{equation}
\G_i=\left\{
\begin{array}{ll}
\langle \J_i, \U_i, \vec{w} \rangle & \text{$|\J_i|$ even} \\
\langle \J_i\cup\{\Delta\}, \U_i\cup\{\{j,\Delta\}|j\in\J_i\},\vec{w} \rangle & \text{$|\J_i|$ odd}. 
\label{eq:Gi}
\end{array}\right.
\end{equation}
In~\eqref{eq:Gi}, the graph is represented by a 3-tuple, with the first element being the vertex set, the second element being the edge set, and the third element being the weight vector.
Parameter $\Delta$ is an auxiliary vertex for odd $|\J_i|$. Without loss of generality, below we focus on odd $|\J_i|$. (All conclusions naturally hold for $|\J_i|$ being even.)
By the definition in~\eqref{eq:Gi}, each UE is corresponding to a vertex. For each pair $\u$ in $\U_i$, there is one edge connecting the two UEs in $\u$, associated with weight $w_{\u}$. We name these as \textit{type-1 edges}. Besides, for each UE $j$ in $\J_i$, there is one extra edge connecting $j$ and the auxiliary vertex $\Delta$, associated with weight $w_{j}$. We name these as \textit{type-2 edges}.
An illustration is given in~\figurename~\ref{fig:matching}. 
\begin{figure}[!t]
\centering
\includegraphics[width=0.7\linewidth]{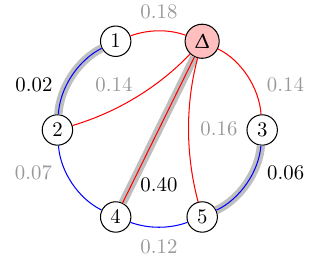}
\caption{The figure shows an example of one cell $i$ with five UEs, i.e., $\J_i=\{1,2,3,4,5\}$. Assume the candidate pair set is $\U_i=\{\{1,2\},\{2,4\},\{4,5\},\{3,5\}\}$. The blue edges are type-1. The red edges are type-2. A matching is a set of edges without common vertices (also called independent edge set) and is a pair selection solution. The maximum matching, as highlighted in the figure, is $\{\{1,2\},\{3,5\},\{4,\Delta\}\}$. \textcolor[rgb]{0,0,0}{Note that two paired UEs in the solution of matching does not necessarily imply that the two share resource via NOMA. For any pair $\u$, if $x_{\u}$ happens to be zero in the solution, then there is no RB allocated in non-orthogonal manner to the pair $\u$ and hence the two UEs in $\u$ are allocated with orthogonal resources.}}
\label{fig:matching}
\end{figure}

The weight $\vec{w}$ is defined as follows, where $T$ is a positive value keeping all weights being positive.
\begin{center}
\begin{tabular}{c|l}
\toprule
Type-1 edge & $w_{\u}=T-\min_{x_{\u}}Z_{\u}(x_{\u},\bm{\rho}_{-i})$ ($\u\in\U_i$)\\
Type-2 edge & $w_{j}=T-d_j/c_j(\bm{\rho}_{-i})$ ($j\in\J_i$) \\
\bottomrule
\end{tabular}
\end{center}

First, we remark that any $\vec{y}_i$ is feasible to~\eqref{eq:minload} if and only if all the pairs 
$\u$ with $y_{\u}=1$ ($\u\in\U_i$) form a matching (or an empty edge set) in $\G_i$. Otherwise, there exists $j$ such that $\sum_{\u\in\V_j}y_{\u}\geq 2$, and~\eqref{eq:minf-y} would be violated. Then, by the definition of weights, minimizing the load $\rho_i$ becomes finding a maximum weighted matching.
\end{proof}

The algorithm \textsc{S-Cell} solving~\eqref{eq:minload} exactly is as follows. Lines~\ref{alg:cell-compute-weight-start}--\ref{alg:cell-compute-weight-end} compute the edge weights of the graph to be constructed. Then we construct the graph in Line~\ref{alg:cell-graph} and compute the maximum matching\footnote{The best known algorithm \cite{4567800} runs on $\G_i$ in $O((|\U_i|+|\J_i|)\sqrt{|\J_i|})$ (odd $|\J_i|$) or $O(|\U_i|\sqrt{|\J_i|})$ (even $|\J_i|$).} $\U_i^{*}$  
in Line~\ref{alg:cell-matching}, which by Theorem~\ref{thm:matching} is the optimal pair selection in cell $i$. Lines~\ref{alg:cell-set-zero}--\ref{alg:cell-loop1-end} assign the obtained solutions to $\langle q_{\oplus\u}^{*},q^{*}_{\ominus\u},x^{*}_{\oplus},x^{*}_{\ominus},x^{*}_{\u} \rangle$ for the pairs in $\U^{*}_i$. The other pairs are not selected and hence their values in $\vec{x}_i^{*}$, $\vec{q}_i^{*}$, and $\vec{y}_i^{*}$ are zeros.
\begin{figure}
\begin{tcolorbox}[boxrule=0.5pt]
\vspace{-0.2cm}
\begin{codebox}
\Procname{$\proc{S-Cell}(\bm{\rho}_{-i})$}
\li \For $\u\in\U_i$ \label{alg:cell-compute-weight-start}
\li \> $\langle q_{\oplus\u},q_{\ominus\u},x_{\oplus},x_{\ominus},x_{\u} \rangle \gets \proc{Split}(\u,\bm{\rho}_{-i})$
\li \> $w_{\u} \gets T-(x_{\oplus}+x_{\ominus}+x_{\u})$  \textcolor[rgb]{0.18,0.53,0.33}{\Comment\small $w_{\u}>0$} 
\li \textbf{end for}
\li \If $|\J_i|$ is odd
\li \> \For $j\in\J_i$
\li \> \> $x_j \gets d_j/c_j(\bm{\rho}_{-i})$
\li \> \> $w_j \gets T-x_j$ \textcolor[rgb]{0.18,0.53,0.33}{\Comment\small $w_{j}>0$} \label{alg:cell-compute-weight-end}
\li \> \textbf{end for}
\li \textbf{end if}
\li Construct $\G_i$ by~\eqref{eq:Gi}  \label{alg:cell-graph}
\li $\U_i^{*} \gets \proc{Maximum-Weighted-Matching}(\G_i)$ \label{alg:cell-matching}
\li $\vec{x}^{*}_i\gets \vec{0}$,~$\vec{q}^{*}_i\gets \vec{0}$,~$\vec{y}^{*}_i\gets \vec{0}$ \label{alg:cell-set-zero}
\li \For $\u\in\U^{*}_i\cap\U_i$ \label{alg:cell-loop2-start}
\li \> $x^{*}_{\u}\gets x_{\u}$, $y^{*}_{\u}\gets 1$
\li \> \For $j\in\u$
\li \> \> $q^{*}_{j\u}\gets q_{j\u}$,~$x^{*}_{j}\gets x_{j}$ 
\li \> \textbf{end for}
\li \textbf{end for}
\li \If $|\J_i|$ is odd
\li \> Find the $\{j,\Delta\}$ in $\U_i^{*}$ and let $x^{*}_j \gets x_j$
\li \textbf{end if}
\label{alg:cell-loop1-end}
\li $\rho_i=\sum_{j\in\J_i}x^{*}_{j}+\sum_{\u\in\U_i}x^{*}_{\u}$ \label{alg:cell-rhoi}
\li \Return $\langle\rho_i,\vec{q}^{*}_i,\vec{x}^{*}_{i},\vec{y}^{*}_{i}\rangle$
\end{codebox}
\end{tcolorbox}
\end{figure}

\textcolor[rgb]{0,0,0}{We remark that in the matching process, if the number of nodes is odd, for the unpaired UE, it is allocated with orthogonal RBs. For two UEs that are paired in the solution of matching, they are not necessarily in non-orthogonal allocation but is up to optimization.}

\section{Multi-cell Load Optimization}
\label{sec:multiple}

This section proposes the algorithmic framework \textsc{M-Cell} for deriving the optimum of~\minf, by analyzing sufficient-and-necessary conditions of optimality and feasibility. 

\subsection{Revisiting Single-cell Load Minimization}
\label{subsec:revisiting}
%
Recall that for single cell optimization, the optimum of~\eqref{eq:minload} of cell $i$ ($i\in\I$) is a function of the load of other cells $\bm{\rho}_{-i}$. By Lemma~\ref{lma:feasible} below, this function is well-defined for any non-negative $\bm{\rho}_{-i}$.
\begin{lemma}
The problem in \eqref{eq:minload} is always feasible.
\label{lma:feasible}
\end{lemma}
\begin{proof}
We select some $\vec{y}_i$ in $\Y_i$ and fix it in~\eqref{eq:minload} ($\Y_i\neq\phi$ by definition). For each pair $\u$, we fix $\vec{q}_{\u}$ to $[p_i/2,p_i/2]^{\top}$. To prove~\eqref{eq:minload} is feasible, we prove the remaining problem is always feasible. Note that, with $\vec{y}_i$ and $\vec{q}_{\u}$ being fixed, \eqref{eq:minload} becomes a linear programming (LP) problem, which is stated below (the equalities are by Lemma~\ref{lma:equality}). 
\begin{equation}
\min_{\vec{x}\geq\vec{0}}\sum_{j\in\J_i}x_j+\sum_{\u\in\U_i}x_{\u},\text{~s.t.~}c_{j}x_{j}+\sum_{\u\in\V_j}c_{j\u}x_{\u}=d_j.	
\label{eq:reformulation}
\end{equation}
By Farkas' lemma, a group of linear constraints in standard form, i.e. $\vec{A}\vec{x}=\vec{b}$ ($\vec{b}\geq\vec{0}$), is feasible with $\vec{x}\geq\vec{0}$ if and only if there does not exist $\vec{v}$ such that $\vec{v}^{\top}\vec{A}\geq\vec{0}^{\top}$ and $\vec{v}^{\top}\vec{b}<0$.
Obviously, there is no $\vec{v}$ with $\vec{d}\geq \vec{0}$ satisfying~\eqref{eq:farkas}.
\begin{equation}
v_{j}\geq 0~(j\in\J_i)\textnormal{ and }\sum_{j\in\J_i}v_{j}d_{j}<0
\label{eq:farkas}
\end{equation}
Hence \eqref{eq:reformulation} is feasible, and the conclusion holds.
\end{proof}

Let $\lambda_i=|\Y_i|$. For each $\vec{y}_i$ in $\Y_i$, we use an integer in $[1,\lambda_i]$ to uniquely index $\vec{y}_i$. 
We refer to all the pair selection solutions in $\Y_i$ as pairing $1$, pairing $2$, \ldots, pairing $\lambda_i$.
Denote by $f_{ik}(\bm{\rho}_{-i})$ the optimum of~\eqref{eq:minload} under pairing $k$ ($1\leq k\leq\lambda_i$), i.e.,
\[
f_{ik}(\bm{\rho}_{-i}) =	\min\limits_{\rho_i,\vec{q}_i,\vec{x}_i} \rho_i~\text{s.t. \eqref{eq:minf-rho}--\eqref{eq:minf-xy} and \eqref{eq:minf-rhoqx} of cell $i$.} 
\]
Let $f_i(\bm{\rho}_{-i})$ be the optimum\footnote{Therefore, $f_i(\bm{\rho}_{-i})$ equals the $\rho_i$ obtained from $\textsc{S-Cell}(\bm{\rho}_{-i})$.} of~\eqref{eq:minload}. Then we have:
\begin{equation}
f_i(\bm{\rho}_{-i})=\min_{k=1,2,\ldots,\lambda_i} f_{ik}(\bm{\rho}_{-i}).
\label{eq:fi}
\end{equation}
Network-wisely, we have:
\begin{equation}
\vec{f}(\bm{\rho})=[f_1(\bm{\rho}_{-1}),f_2(\bm{\rho}_{-2}),\ldots,f_n(\bm{\rho}_{-n})].
\label{eq:f}	
\end{equation}
The following theorem reveals a key property of $\vec{f}(\bm{\rho})$.
\begin{theorem}
$\vec{f}(\bm{\rho})$ is an SIF, i.e.
the
following properties hold: 
\begin{enumerate} 
\item (Scalability)
$\alpha\vec{f}(\bm{\rho})>\vec{f}(\alpha\bm{\rho}),~
\bm{\rho}\geq\vec{0},~\alpha>1$.
\item (Monotonicity) $\vec{f}(\bm{\rho})\geq \vec{f}(\bm{\rho}')$,
$\bm{\rho} \geq \bm{\rho}'$, $\bm{\rho},\bm{\rho}'\geq\vec{0}$.
\end{enumerate} 
\label{thm:sif}
\end{theorem}
\begin{proof}
We first prove monotonicity and scalability 
for $f_{ik}(\bm{\rho}_{-i})$ 
($i\in\I$, $k=1,2,\ldots,\lambda_i$).
For monotonicity, we prove 
that $f_{ik}(\bm{\rho}'_{-i})\leq f_{ik}(\bm{\rho}_{-i})$ for $\bm{\rho}'_{-i}\leq\bm{\rho}_{-i}$ as follows. 
Given any non-negative $\bm{\rho}_{-i}$, we replace $\bm{\rho}_{-i}$ with $\bm{\rho}'_{-i}$. 
Note that $c_{j\u}(\bm{\rho}'_{-i})\geq c_{j\u}(\bm{\rho}_{-i})$. Thus the replacement makes the solution space of \eqref{eq:minf-d} larger, 
and the optimum with $\bm{\rho}'_{-i}$ is no larger than that with $\bm{\rho}_{-i}$. Therefore $f_{ik}(\bm{\rho}'_{-i})\leq f_{ik}(\bm{\rho}_{-i})$. For scalability, we prove that $f_{ik}(\alpha\bm{\rho}_{-i})\leq \alpha f_{ik}(\bm{\rho}_{-i})$ for $\alpha>1$ and non-negative $\bm{\rho}_{-i}$ as follows. 
Denote the optimal solution of $f_{ik}(\bm{\rho}_{-i})$ by $\langle\rho''_i,\vec{q}_i'',\vec{x}_i''\rangle$. We have $f_{ik}(\bm{\rho}_{-i})=\rho''_i$. 
Due to that $1/c_{j\u}(\vec{q}_i,\bm{\rho}_{-i})$ and $1/c_j(\bm{\rho}_{-i})$ are strictly concave in $\bm{\rho}_{-i}$, the two inequalities
\begin{equation}
\frac{1}{c_{j}(\alpha\bm{\rho}_{-i})} < \frac{\alpha}{c_{j}(\bm{\rho}_{-i})}\text{, }\frac{1}{c_{j\u}(\vec{q}_i,\alpha\bm{\rho}_{-i})} < \frac{\alpha}{c_{j\u}(\vec{q}_i,\bm{\rho}_{-i})}
\label{eq:concave_c}
\end{equation}
hold for $\alpha>1$. Consider the following minimization problem~\ref{eq:alpha_relax}, with $\vec{y}_i$ being fixed to pairing $k$. 
\begin{subequations}
\begin{alignat}{2}
&\min\limits_{\rho_i,\vec{q}_i,\vec{x}_i\geq\vec{0}} \rho_i \\
\text{s.t.} \quad & \text{\eqref{eq:minf-rho}, \eqref{eq:minf-p} and \eqref{eq:minf-xy} of cell $i$}\\
&  c_j(\bm{\rho}_{-i})x_j +  \sum_{\u\in\V_j}c_{j\u}(\vec{q},\bm{\rho}_{-i})x_{\u}\geq \alpha d_j,~j\in\J_i \label{eq:alpha_relax-d}
\end{alignat}
\label{eq:alpha_relax}
\end{subequations}
\!\!Note that $\langle\alpha\rho''_i,\vec{q}''_i,\alpha\vec{x}''_i\rangle$ is feasible to~\eqref{eq:alpha_relax}, 
with the objective value being $\alpha f_{ik}(\bm{\rho}_{-i})$. 
Hence the optimum of~\eqref{eq:alpha_relax} is no more than $\alpha f_{ik}(\bm{\rho}_{-i})$. 
For $f_{ik}(\alpha\bm{\rho}_{-i})$, note that the corresponding optimization problem only differs with~\eqref{eq:alpha_relax} in~\eqref{eq:alpha_relax-d}. Instead of~\eqref{eq:alpha_relax-d}, in $f_{ik}(\alpha\bm{\rho}_{-i})$ we have:
\begin{equation}
c_j(\alpha\bm{\rho}_{-i})x_j + \sum_{\u\in\V_j} c_{j\u}(\vec{q},\alpha\bm{\rho}_{-i})\geq d_j,~j\in\J_i
\label{eq:c_j_alpha_rho}
\end{equation}
By Lemma~\ref{lma:equality}, \eqref{eq:alpha_relax-d} is equality at the optimum.
Then by~\eqref{eq:concave_c}, for any solution of~\eqref{eq:alpha_relax}, using it for the optimization problem associated with $f_{ik}(\alpha\bm{\rho}_{-i})$ makes \eqref{eq:c_j_alpha_rho} an inequality. (This is because by~\eqref{eq:concave_c} we obtain $c_j(\bm{\rho}_{-i})/\alpha < c_j(\alpha\bm{\rho}_{-i})$ and $c_{j\u}(\vec{q}_i,\bm{\rho}_{-i})/\alpha < c_{j\u}(\vec{q}_i,\alpha\bm{\rho}_{-i})$).
Therefore the problem for $f_{ik}(\alpha\bm{\rho}_{-i})$ has a lower optimum than \eqref{eq:alpha_relax}. Further, the optimum is lower than $\alpha f_{ik}(\bm{\rho}_{-i})$. Hence $f_{ik}(\alpha\bm{\rho}_{-i})<\alpha f_{ik}(\bm{\rho}_{-i})$.

We then allow $k$ to be variable and consider $f_{i}(\bm{\rho}_{-i})$ ($i\in\I$). For $\bm{\rho}'_{-i}\leq\bm{\rho}_{-i}$ we have
\[
f_i(\bm{\rho}'_{-i})=\min_{k}f_{ik}(\bm{\rho}'_{-i})\leq\min_{k}f_{ik}(\bm{\rho}_{-i})=f_i(\bm{\rho}_{-i})
\]
and for $\alpha>1$ we have
\[
f_i(\alpha\bm{\rho}_{-i})=\min_{k}f_{ik}(\alpha \bm{\rho}_{-i})< \alpha \min_{k}f_{ik}(\bm{\rho}_{-i})=\alpha f_{i}(\bm{\rho}_{-i})
\]
Hence the conclusion.
\end{proof}

Given $\bm{\rho}$, denote by $\vec{f}^{k}$ ($k> 1$) the function composition of $\vec{f}(\vec{f}^{k-1}(\bm{\rho}))$ (with $\vec{f}^{0}(\bm{\rho})=\bm{\rho}$). 
Lemma~\ref{lma:lim} holds by~\cite{414651}. 
\begin{lemma}
If $\lim_{k\rightarrow\infty}\vec{f}^{k}(\bm{\rho})$ exists, then it exists uniquely for any $\bm{\rho}\geq\vec{0}$.
\label{lma:lim}
\end{lemma}

\subsection{Optimality and Feasibility}
\label{subsec:nopower}

Based on Theorem~\ref{thm:sif}, we derive sufficient-and-necessary conditions for \minf~in terms of its feasibility and optimality. For any load $\bm{\rho}$, we say that a load $\bm{\rho}$ is achievable if and only if there exist $\vec{q}$, $\vec{x}$, and $\vec{y}$ such that the solution $\langle \bm{\rho},\vec{q},\vec{x},\vec{y}\rangle$ is feasible to \minf.
\begin{lemma}
For any $\bm{\rho}\geq\vec{0}$, if there exists $i\in\I$ such that
$\rho_i<f_i(\bm{\rho}_{-i})$, then $\bm{\rho}$ is not achievable in \minf.
\label{lma:not_feasible}
\end{lemma}
\begin{proof}
Let $\rho'_i = f_i(\bm{\rho}_{-i})$. By the definition of $f_i$, $\rho'_i$ is
the minimum value satisfying
\eqref{eq:minf-rho}\textnormal{--}\eqref{eq:minf-integer} under
$\bm{\rho}_{-i}$. Therefore any $\rho_i$ with $\rho_i<\rho'_i$ is not achievable with
constraints~\eqref{eq:minf-rho}\textnormal{--}\eqref{eq:minf-integer}. Hence the conclusion.
\end{proof}
\begin{theorem}
In \minf,
$\bm{\rho}$ ($\bm{\rho}\leq\bar{\bm{\rho}}$) is achievable if
and only if $\vec{f}(\bm{\rho})$ is achievable and
$\bm{\rho}\geq \vec{f}(\bm{\rho})$.
\label{thm:feasibility}
\end{theorem}
\begin{proof}
By the inverse proposition of
Lemma~\ref{lma:not_feasible}, an achievable $\bm{\rho}$ always satisfies $\bm{\rho}\geq\vec{f}(\bm{\rho})$. The necessity is proved as follows. Suppose $\bm{\rho}$ is achievable for \minf. Consider using $\vec{f}(\bm{\rho})$ as another solution (together with the $\langle \vec{q},\vec{x},\vec{y} \rangle$ obtained when computing $\vec{f}(\bm{\rho})$). Then $\vec{f}(\bm{\rho})$ satisfies~\eqref{eq:minf-rhobar}. Also, $\vec{f}(\bm{\rho})$ together with its $\langle \vec{q},\vec{x},\vec{y} \rangle$ fulfills \eqref{eq:minf-rho}\textnormal{--}\eqref{eq:minf-integer} by the definition of $\vec{f}(\bm{\rho})$. Thus, $\vec{f}(\bm{\rho})$ is achievable. 


For the sufficiency, note that the achievability of $\vec{f}(\bm{\rho})$ implies that $\bm{\rho}$ along with $\langle \vec{q},\vec{x},\vec{y} \rangle$ obtained by solving $\vec{f}(\bm{\rho})$ satisfies \eqref{eq:minf-rho}\textnormal{--}\eqref{eq:minf-integer}.  Combined with the precondition $\rho_i\leq\bar{\rho}$ ($i\in\I$), the load $\bm{\rho}$ is feasible to \eqref{eq:minf-rhobar}\textnormal{--}\eqref{eq:minf-integer} (and thus achievable in \minf). Hence the conclusion.
\end{proof}

Theorem~\ref{thm:feasibility} provides an effective method for improving any sub-optimal solution to \minf. \textcolor[rgb]{0,0,0}{\textit{For any achievable $\bm{\rho}$, evaluating $\vec{f}(\bm{\rho})$ always yields a better solution}}\footnote{Rigorously, Theorem~\ref{thm:feasibility} implies that the new solution is not worse. In fact it is guaranteed to be strictly better (with strictly monotonic $F(\bm{\rho})$) unless the old one is already optimal. A proof can be easily derived based on Theorem~\ref{thm:optimality}.}. 
\textcolor[rgb]{0,0,0}{This conclusion is based on Theorem~\ref{thm:feasibility}: Suppose $\bm{\rho}$ ($\bm{\rho}\geq\vec{0}$) is the current cell load, and let $\bm{\rho}'$ be the function value evaluated at $\bm{\rho}$, i.e. $\bm{\rho}'=\vec{f}(\bm{\rho})$. By Theorem \ref{thm:feasibility}, we always have $\bm{\rho}'\leq\bm{\rho}$.}

\textcolor[rgb]{0,0,0}{
Recall that $F(\bm{\rho})$ is the objective function of the problem \textsc{MinF}.
Theorem~\ref{thm:optimality} below states that, the fixed point of $\vec{f}(\bm{\rho})$ (along with $\langle \vec{q},\vec{x},\vec{y}\rangle$ obtained when computing $\vec{f}(\bm{\rho})$) is optimal to \minf. }

\begin{theorem} 
Load $\bm{\rho}^{*}$ is the optimum of 
\minf~if (and only if when $F(\bm{\rho})$ is strictly monotonic) $\bm{\rho}^{*} = \vec{f}(\bm{\rho}^{*})\leq\bar{\bm{\rho}}$.
\label{thm:optimality}
\end{theorem}
\begin{proof}
(Necessity) If $\bm{\rho}^{*}$ is optimal (and thus feasible), then obviously we have $\bm{\rho}^{*}\leq\bar{\bm{\rho}}$.
By
Theorem~\ref{thm:feasibility}, $\vec{f}(\bm{\rho}^{*})$ 
is also feasible and 
$\vec{f}(\bm{\rho}^{*})\leq\bm{\rho}^{*}$. By successively applying Theorem~\ref{thm:feasibility}, $\vec{f}^{k}(\bm{\rho}^{*})$ for any $k\geq 1$ is a 
feasible solution and $\vec{f}^{k}(\bm{\rho}^{*})\leq\vec{f}^{k-1}(\bm{\rho}^{*})$. Let $\bm{\rho}'=\lim_{k\rightarrow\infty}\vec{f}^{k}(\bm{\rho}^{*})$.
Then $\bm{\rho}'\leq \bm{\rho}^{*}$ holds by the above derivation. In addition, note that $\bm{\rho}'$ is a feasible solution as well. By that $\bm{\rho}^{*}$ is optimal for \minf, we have $\bm{\rho}'=\bm{\rho}^{*}$, otherwise $\bm{\rho}'$ would lead to a better objective value in \minf~than $\bm{\rho}^{*}$. Hence $\bm{\rho}^{*}=\lim_{k\rightarrow\infty}\vec{f}^{k}(\bm{\rho}^{*})$, i.e. $\bm{\rho}^{*}=\vec{f}(\bm{\rho}^{*})$.

(Sufficiency) By Theorem~\ref{thm:feasibility}, for any feasible $\bm{\rho}$, $\lim_{k\rightarrow\infty}\vec{f}^{k}(\bm{\rho})$ is feasible and $\lim_{k\rightarrow\infty}\vec{f}^{k}(\bm{\rho})\leq\bm{\rho}$ holds. By Lemma~\ref{lma:lim}, the limit remains for any $\bm{\rho}\geq\vec{0}$, and thus $\lim_{k\rightarrow\infty}\vec{f}^{k}(\bm{\rho})=\lim_{k\rightarrow\infty}\vec{f}^{k}(\bm{\rho}^{*})$. Since $\bm{\rho}^{*}=\vec{f}(\bm{\rho}^{*})$, we have $\bm{\rho}^{*}=\lim_{k\rightarrow\infty}\vec{f}^{k}(\bm{\rho}^{*})$. Thus $\bm{\rho}^{*}\leq\bm{\rho}$ for any feasible $\bm{\rho}$, meaning that $\bm{\rho}^{*}$ is optimal for \minf.

Hence the conclusion.
\end{proof}



\subsection{The Algorithmic Framework}

Starting from any non-negative $\bm{\rho}^{(0)}$, we compute $\lim_{k\rightarrow\infty}\vec{f}^{k}(\bm{\rho})$ iteratively. During each iteration, $n$ problems in~\eqref{eq:minload} for $i\in\I$ are solved. The convergence is guaranteed by Lemma~\ref{lma:lim}. At the convergence, by Theorem~\ref{thm:optimality}, the optimum is reached. Note that once $\bm{\rho}^{(k)}$ is feasible for any $k\geq 0$, then by Theorem~\ref{thm:feasibility}, all $\bm{\rho}^{(k+1)},\bm{\rho}^{(k+2)},\ldots$ are feasible as well. One can terminate prematurely to obtain a sub-optimal solution with less computation. \textsc{M-Cell} is outlined below. 
\begin{tcolorbox}[boxrule=0.5pt]
\vspace{-0.2cm}
\begin{codebox}
\Procname{$\proc{M-Cell}(\bm{\rho}^{(0)},\epsilon)$}
\li $k\gets 0$ \label{alg:network-k}

\li \Repeat 
\li $k \gets k + 1$
\li  \For $i \in \I$
\li  \> $\langle\rho_i^{(k)},\vec{q}_i^{(k)},\vec{x}^{(k)}_i,\vec{y}_i^{(k)}\rangle \gets \proc{S-Cell}(\bm{\rho}_{-i}^{(k-1)})$ \label{alg:network-single-cell}
\li \textbf{end for}
\li \Until $\norm{\bm{\rho}^{(k)}-\bm{\rho}^{(k-1)}}_{\infty}\leq\epsilon$ \label{alg:network-epsilon}
\li \textcolor[rgb]{0,0,0}{\If $\rho_i^{(k)}>\bar{\rho}$ for some $i$ ($i\in\I$)} \label{alg:network-feasibility-1}
\li \> \textcolor[rgb]{0,0,0}{\textsc{MinF} is infeasible}\label{alg:network-feasibility-2}
\li \textbf{end if}
\li \Return $\langle\bm{\rho}^{(k)},\vec{q}^{(k)},\vec{x}^{(k)},\vec{y}^{(k)}\rangle$\label{alg:network-return}
\end{codebox}
\end{tcolorbox}

\textsc{M-Cell} applies fixed point iterations using $\vec{f}(\bm{\rho})$. 
The convergence of fixed point iterations on $\vec{f}(\bm{\rho})$ is linear~\cite{lemmens_nussbaum_2012}. \textcolor[rgb]{0,0,0}{The feasibility check is done by Lines \ref{alg:network-feasibility-1} and \ref{alg:network-feasibility-2}. The infeasibility of \textsc{MinF} implies that at least one cell will be overloaded for meeting user demands. If this happens, we know for sure that the user demands cannot be satisfied.}
We remark that all the conclusions derived in this section are independent of the implementation of \textsc{S-Cell} in Line \ref{alg:network-single-cell}. As long as the sub-routine \textsc{S-Cell} yields the optimal solution to~\eqref{eq:minload}, \textsc{M-Cell} achieves the optimum of \minf\footnote{\textcolor[rgb]{0,0,0}{With filtered $\U$, the proposed \textsc{M-Cell} is proved to converge to the global optimum of \textsc{MinF}. Without filtered $\U$ (or for any possible candidate pairs set $\U$), \textsc{M-Cell} is still applicable to \textsc{MinF} though there is no theoretical guarantee of convergence or optimality, as the decoding order for each pair may change in the iteration process.}}. 
Besides, \textsc{M-Cell} possesses the optimality for \textsc{MinF} with any objective function that is monotonically (but not necessarily strictly monotonic) increasing in each element of $\bm{\rho}$. 
These two properties make \textsc{M-Cell} an algorithmic framework. To our knowledge, the most efficient \textsc{S-Cell} is what we derived in Section~\ref{sec:single}.

\pgfplotsset{compat=1.11,
        /pgfplots/ybar legend/.style={
        /pgfplots/legend image code/.code={%
        \draw[##1,/tikz/.cd,bar width=3pt,yshift=-0.2em,bar shift=0pt]
                plot coordinates {(0cm,0.8em)};},
},
}
\begin{figure}[t!] 
\centering 
\begin{tikzpicture}
\begin{axis}[
	xlabel={Normalized demand $d$},
	ylabel={Total load},
	label style = {font=\fontsize{9pt}{10pt}\selectfont},
	legend cell align={left},
	legend pos = north west,
	legend style = {font=\fontsize{8pt}{10pt}\selectfont},
	axis background/.style={fill=white},
minor x tick num=4,
minor y tick num=4,
major tick length=0.15cm,
minor tick length=0.075cm,
tick style={semithick,color=black},
	height=0.667\linewidth,
	width=\linewidth,
]
	
\addplot [smooth, mark=*, color=black] coordinates {
	(0.2, 1.26603)
	(0.4, 3.25773)
	(0.6, 5.98824)
	(0.8, 9.64595)
	(1.0, 14.5622)
};

\addplot [smooth, mark=square, color=brown] coordinates {
	(0.2, 1.26603)
	(0.4, 3.07362)
	(0.6, 5.56959)
	(0.8, 8.85518)
	(1.0, 13.2036)
};

\addplot [smooth, mark=triangle*, color=blue] coordinates {
	(0.2, 1.15191)
	(0.4, 2.87725)
	(0.6, 5.16167)
	(0.8, 8.14019)
	(1.0, 12.0569)
};

\addplot [smooth, mark= diamond, color=red] coordinates{
	(0.2, 1.21471)
	(0.4, 2.34416)
	(0.6, 4.23353)
	(0.8, 6.72999)
	(1.0, 10.0478)
};

\legend{OMA$_{\text{(Opt)}}$, SP$_{\text{(Uni)}}$, SP$_{\text{(FTPC)}}$, NOMA$_{\text{(Opt)}}$}
\end{axis}
\end{tikzpicture}\\~\\

\begin{tikzpicture}
\begin{axis}[
	label style = {font=\fontsize{9pt}{10pt}\selectfont},
	legend cell align={left},
	legend pos = north west,
	legend style = {font=\fontsize{8pt}{10pt}\selectfont},
	axis background/.style={fill=white},
	xlabel={Normalized demand $d$},
	ylabel={Max load},
minor x tick num=4,
minor y tick num=4,
major tick length=0.15cm,
minor tick length=0.075cm,
tick style={semithick,color=black},
	height=0.667\linewidth,
	width=\linewidth,
]
	
\addplot [smooth, mark=*, color=black] coordinates {
	(0.2, 0.0769911)
	(0.4, 0.205131)
	(0.6, 0.388603)
	(0.8, 0.644085)
	(1.0, 1.0)
};

\addplot [smooth, mark=square, color=brown] coordinates {
	(0.2, 0.0736585)
	(0.4, 0.192651)
	(0.6, 0.359418)
	(0.8, 0.587716)
	(1.0, 0.901219)
};

\addplot [smooth, mark=triangle*, color=blue] coordinates {
	(0.2, 0.0701859)
	(0.4, 0.181201)
	(0.6, 0.334799)
	(0.8, 0.543253)
	(1.0, 0.827574)
};

\addplot [smooth, mark= diamond, color=red] coordinates{
	(0.2, 0.0581377)
	(0.4, 0.150397)
	(0.6, 0.279505)
	(0.8, 0.456336)
	(1.0, 0.699261)
};
\legend{OMA$_{\text{(Opt)}}$, SP$_{\text{(Uni)}}$, SP$_{\text{(FTPC)}}$, NOMA$_{\text{(Opt)}}$}
\end{axis}
\end{tikzpicture}   
\caption{This figure illustrates the total and maximum load in function of normalized demand. At $d=1.0$, the network reaches its resource limit such that any larger demand cannot be satisfied by OMA$_{\text{(Opt)}}$. SP$_{\text{(Uni)}}$ and SP$_{\text{(FTPC)}}$ are two sub-optimal NOMA power allocation schemes, for which, pair selections are optimally computed.}
\label{fig:load} 
\end{figure}
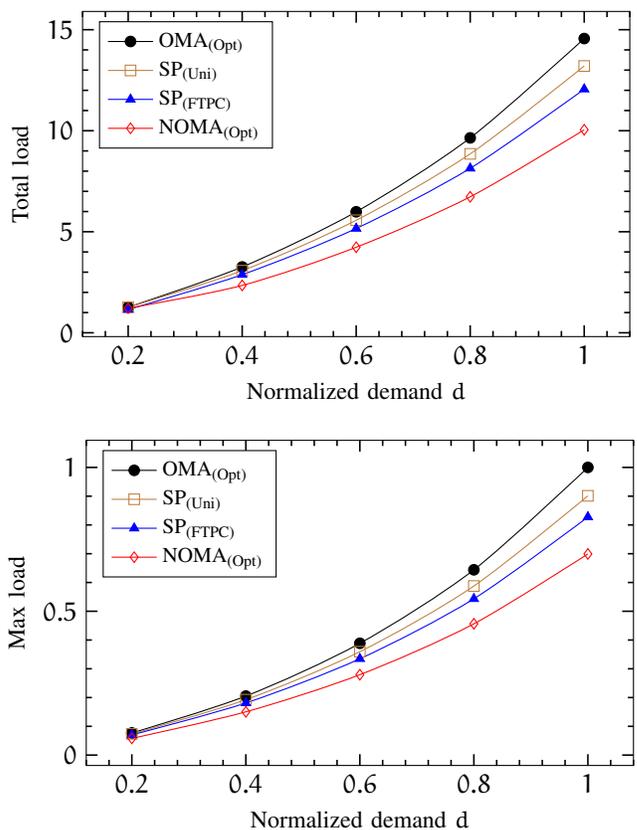

For a cell $i$ ($i\in\I$), given the information of other cells' load $\bm{\rho}_{-i}$, solving $f_i(\bm{\rho})$ is based on local information, making \textsc{M-Cell} suitable to run in a distributed manner. 
A cell can maintain the information of a subset of cells (e.g., the surrounding cells) having major significance in terms of interference, and exchange the information with other cells periodically, which can be implemented via the LTE X2 interface. 
The technique called ``asynchronous fixed-point iterations''~\cite{414651} can be used. 

\begin{figure*}[!t]
\begin{center}
\setlength{\tabcolsep}{0.2em}
\subfigure[PA$_{\text{(B-W)}}$ \label{fig:pairing-best-worst}]{\includegraphics[width=0.245\textwidth]{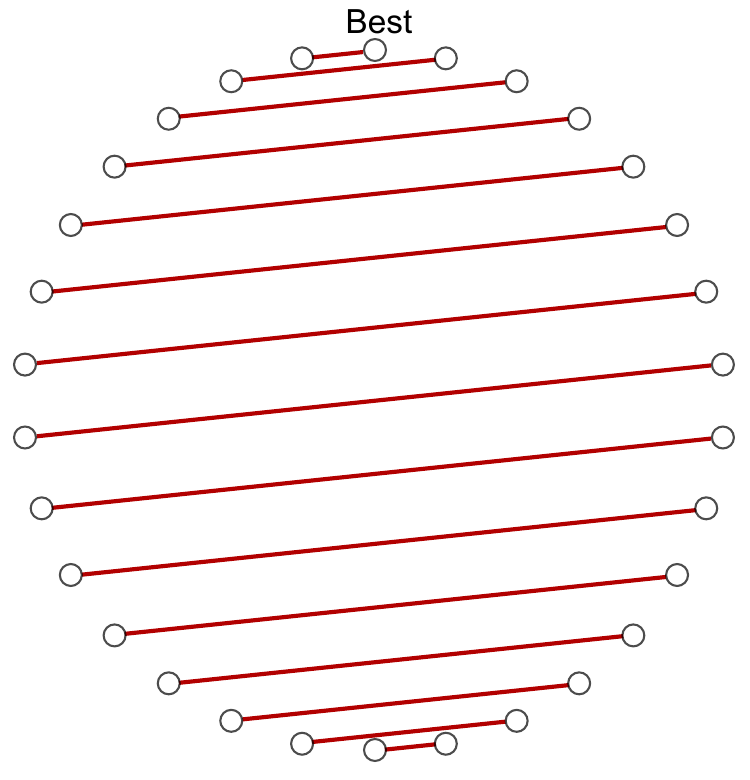}} 
\subfigure[PA$_{\text{(B-SB)}}$ \label{fig:pairing-neighbored}]{\includegraphics[width=0.245\textwidth]{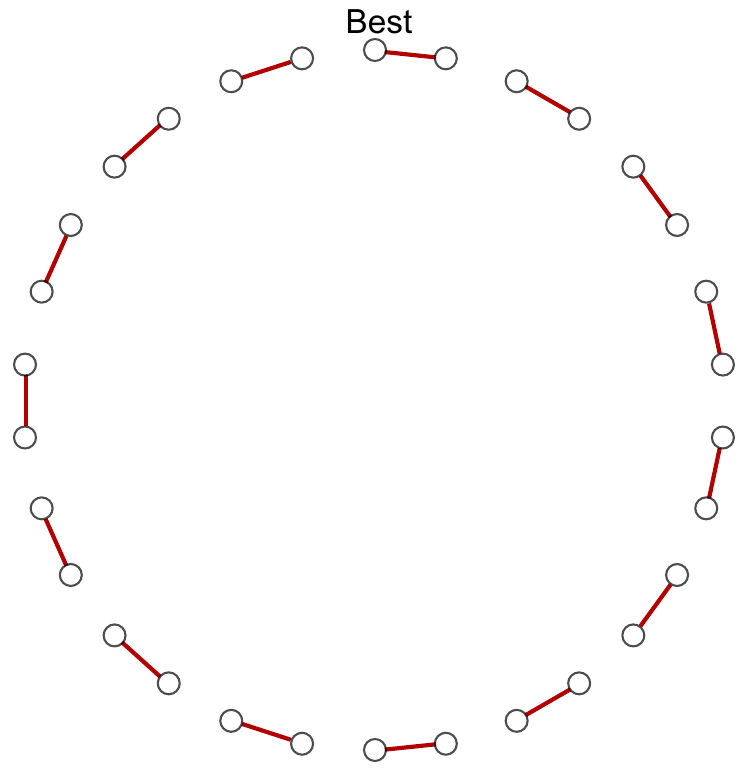}} 
\subfigure[Optimal pairing in filtered $\U$ \label{fig:pairing-filtered}]{\includegraphics[width=0.245\textwidth]{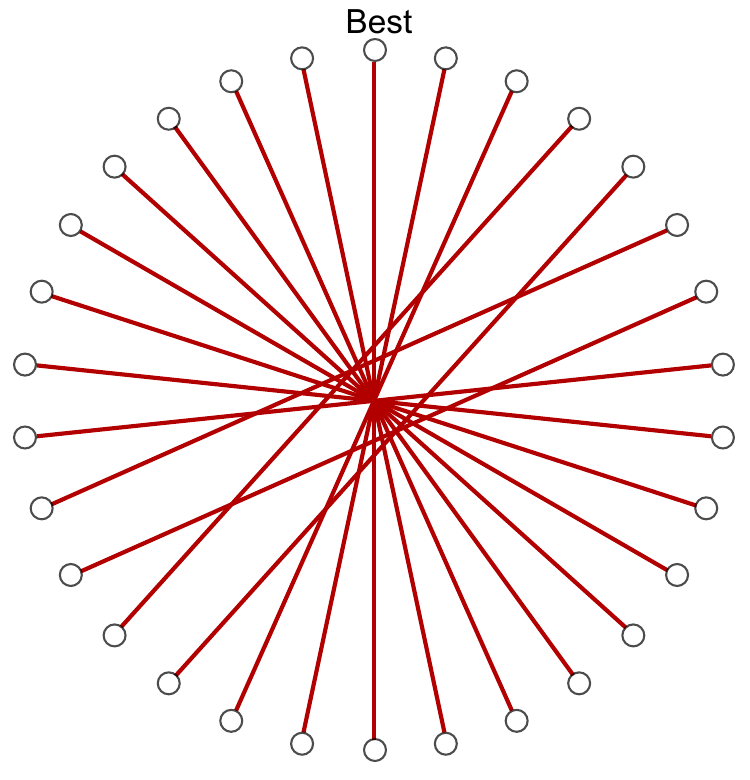}}
\subfigure[Pairing in non-filtered $\U$  \label{fig:pairing-non-filtered}]{\includegraphics[width=0.245\textwidth]{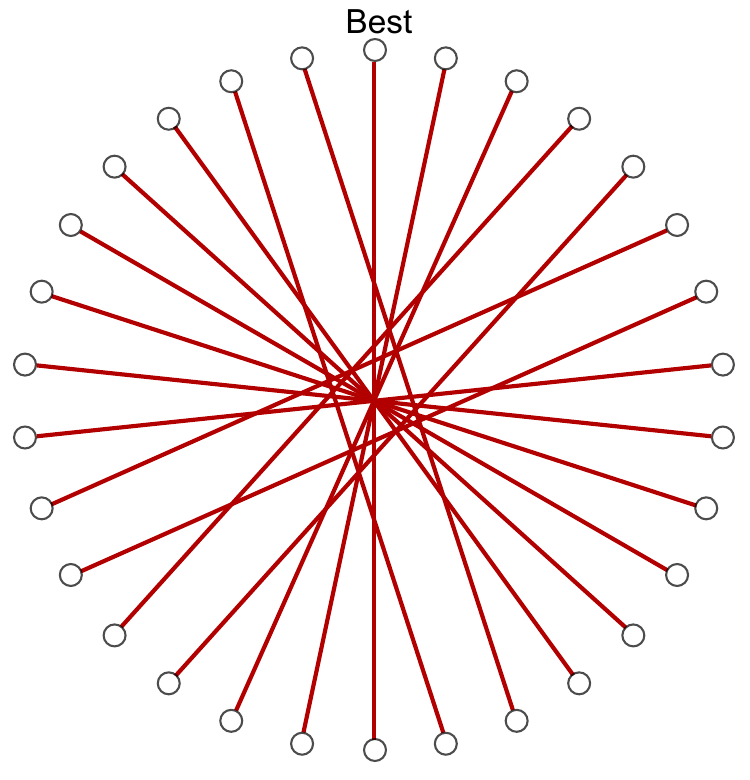}} 
\end{center}
\caption{This figure illustrates pair selection in a typical cell of $30$ UEs. The UEs are represented by the vertices on a circle. 
The UE marked ``\textsf{\scriptsize Best}'' at the top position has the best channel condition. 
The UEs are arranged clock-wisely in the descending order of channel conditions.
The edges are selected pairs. Each subfigure represents one pairing method. \figurename~\ref{fig:pairing-best-worst} and \figurename~\ref{fig:pairing-neighbored} show PA$_{\text{(B-W)}}$ and PA$_{\text{(B-SB)}}$, respectively. \figurename~\ref{fig:pairing-filtered} shows optimal pairing after filtered $\U$. \figurename~\ref{fig:pairing-non-filtered} shows the pairing solution obtained with non-filtered $\U$.}
\label{fig:pairing}
\end{figure*}
\pgfplotsset{compat=1.11,
        /pgfplots/ybar legend/.style={
        /pgfplots/legend image code/.code={%
        \draw[##1,/tikz/.cd,bar width=3pt,yshift=-0.2em,bar shift=0pt]
                plot coordinates {(0cm,0.8em)};},
},
}
\begin{figure*}[t!]
\centering
\begin{tikzpicture}
\begin{axis}[
	font=\fontsize{8pt}{10pt}\selectfont,
    ybar,
    legend style={at={(0.28,0.9)},
      anchor=north,legend columns=-1},
    ylabel={Cell load},
    xlabel={Cell index},
    symbolic x coords={1,2,3,4,5,6,7,8,9,10,11,12,13,14,15,16,17,18,19},
    xtick=data,
    	height=4cm,
	width=\textwidth,
	bar width = 0.2cm,
	xtick align=inside,
	minor y tick num=4,
	major tick length=0.15cm,
	minor tick length=0.075cm,
	tick style={semithick,color=black},
    ]
\addplot  coordinates {(1, 0.383168) (2, 0.440131) (3, 0.462814) (4, 0.469713) (5, 0.48283) (6,0.486942) (7, 0.494808) (8, 0.50069) (9, 0.506148) 
                     (10,0.511738) (11, 0.52642) (12, 0.536375) (13, 0.541067) (14, 0.544915) (15,0.555893) (16, 0.582443) (17, 0.651261) (18,0.65367) (19,0.696861)};
\addplot  coordinates {(1, 0.383346) (2, 0.440405) (3, 0.463309) (4, 0.47017) (5, 0.48307) (6, 0.487386) (7, 0.49505) (8, 0.501243) (9, 0.506527) 
                     (10, 0.512171) (11, 0.527345) (12, 0.537577) (13, 0.541425) (14, 0.546102) (15, 0.556359) (16, 0.587553) (17, 0.652278) (18, 0.65723) (19, 0.699261)};

\legend{$|\U|= {30 \choose 2}\times 19=8265$~~~~~~,$|\U|=5779$ ($\U$ filtered by Lemma~\ref{rmk:decoding})}
\end{axis}
\end{tikzpicture}
\caption{This figure shows the optimized load levels of all $19$ cells. The cells are numbered in an ascending order of loads. The blue bars show the computed cell load under $\U$ composed of all ${30 \choose 2}\times 19$ pairs. The red bars show the minimum cell load with pairs satisfying Lemma~\ref{rmk:decoding}. }
\label{fig:cell-load}
\end{figure*}
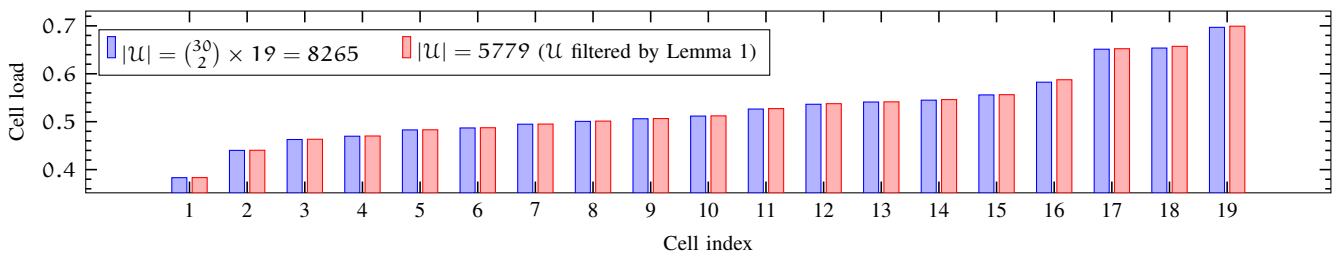
\pgfplotsset{compat=1.11,
        /pgfplots/ybar legend/.style={
        /pgfplots/legend image code/.code={%
        \draw[##1,/tikz/.cd,bar width=3pt,yshift=-0.2em,bar shift=0pt]
                plot coordinates {(0cm,0.8em)};},
},
}
\begin{figure}[!t]
\centering
\begin{tikzpicture}
\begin{axis}[
	label style = {font=\fontsize{9pt}{10pt}\selectfont},
	legend style = {font=\fontsize{8pt}{10pt}\selectfont},
    ybar,
    legend style={at={(0.2,0.96)},
      anchor=north,legend columns=-1},
    ylabel={Average load},
    xlabel={Normalized demand $d$},
    symbolic x coords={0.2, 0.4, 0.6, 0.8, 1.0},
    xtick=data,
	minor x tick num=4,
	minor y tick num=4,
	major tick length=0.15cm,
	minor tick length=0.075cm,
	tick style={semithick,color=black},
	height=0.667\linewidth,
	width=\linewidth,
    bar width = 0.15cm,
    legend entries = {OMA$_{\text{(Opt)}}$, PA$_{\text{(B-W)}}$-SP$_{\text{(Uni)}}$,  PA$_{\text{(B-W)}}$-SP$_{\text{(FTPC)}}$, PA$_{\text{(B-W)}}$-SP$_{\text{(Opt)}}$, NOMA$_{\text{(Opt)}}$},
    legend columns = 1,
    legend cell align={left},
    xtick align=inside,
    ]
\addplot  coordinates {(0.2, 0.0666332) (0.4, 0.171459) (0.6, 0.315171)  (0.8, 0.507682) (1.0, 0.766432) };
\addplot  coordinates {(0.2, 0.0642668) (0.4, 0.163159) (0.6, 0.296563)  (0.8, 0.472869) (1.0, 0.707021) };
\addplot  coordinates {(0.2, 0.0606416) (0.4, 0.152161) (0.6, 0.274757)  (0.8, 0.436202) (1.0, 0.650116) };
\addplot  coordinates {(0.2, 0.0516132) (0.4, 0.130956) (0.6, 0.238668)  (0.8, 0.381958) (1.0, 0.573205) };
\addplot  coordinates {(0.2, 0.0493522) (0.4, 0.123377) (0.6, 0.222817)  (0.8, 0.35421) (1.0, 0.528832) };
\end{axis}
\end{tikzpicture}
\caption{This figure evaluates PA$_{\text{(B-W)}}$, combined with three power split schemes, SP$_{\text{(Uni)}}$, SP$_{\text{(FTPC)}}$, and SP$_{\text{(Opt)}}$. In SP$_{\text{(Opt)}}$, the power split is optimal for each pair. OMA$_{\text{(Opt)}}$ and NOMA$_{\text{(Opt)}}$ are baselines. }
\label{fig:Best-Worst}
\end{figure}
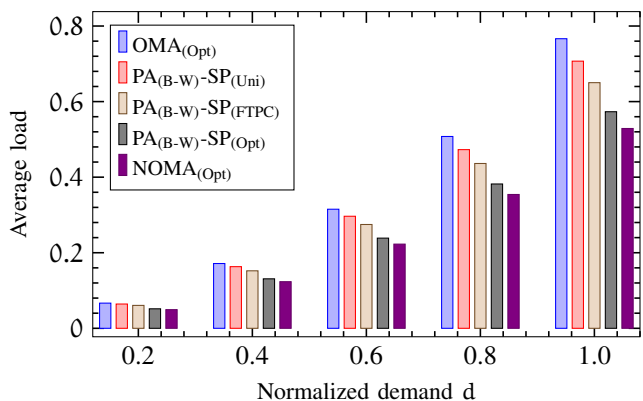

\pgfplotsset{compat=1.11,
        /pgfplots/ybar legend/.style={
        /pgfplots/legend image code/.code={%
        \draw[##1,/tikz/.cd,bar width=3pt,yshift=-0.2em,bar shift=0pt]
                plot coordinates {(0cm,0.8em)};},
},
}
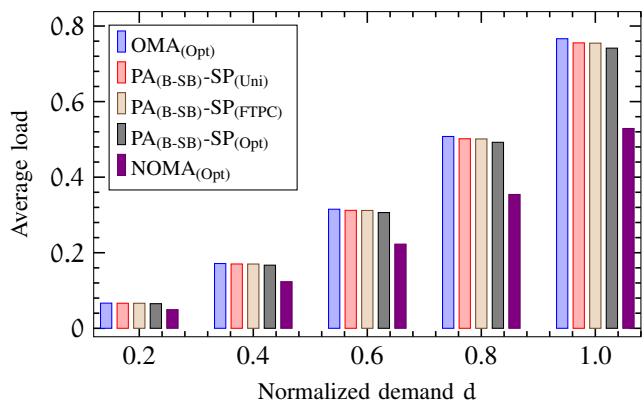
\begin{figure}[!t]
\centering
\begin{tikzpicture}
\begin{axis}[
	label style = {font=\fontsize{9pt}{10pt}\selectfont},
	legend style = {font=\fontsize{8pt}{10pt}\selectfont},
    ybar,
    legend style={at={(0.2,0.96)},
      anchor=north,legend columns=-1},
    ylabel={Average load},
    xlabel={Normalized demand $d$},
    symbolic x coords={0.2, 0.4, 0.6, 0.8, 1.0},
    xtick=data,
	minor x tick num=4,
	minor y tick num=4,
	major tick length=0.15cm,
	minor tick length=0.075cm,
	tick style={semithick,color=black},
	height=0.667\linewidth,
	width=\linewidth,
    bar width = 0.15cm,
    legend entries = {OMA$_{\text{(Opt)}}$, PA$_{\text{(B-SB)}}$-SP$_{\text{(Uni)}}$,  PA$_{\text{(B-SB)}}$-SP$_{\text{(FTPC)}}$, PA$_{\text{(B-SB)}}$-SP$_{\text{(Opt)}}$, NOMA$_{\text{(Opt)}}$},
    legend columns = 1,
    legend cell align={left},
    xtick align=inside,
    ]
\addplot  coordinates {(0.2, 0.0666332) (0.4, 0.171459) (0.6, 0.315171)  (0.8, 0.507682) (1.0, 0.766432) };
\addplot  coordinates {(0.2, 0.0663363) (0.4, 0.170257) (0.6, 0.312196)  (0.8, 0.501665) (1.0, 0.755479) };
\addplot  coordinates {(0.2, 0.0663068) (0.4, 0.170148) (0.6, 0.311947)  (0.8, 0.501193) (1.0, 0.754668) };
\addplot  coordinates {(0.2, 0.0651647) (0.4, 0.167084) (0.6, 0.306346)  (0.8, 0.492281) (1.0, 0.741489) };
\addplot  coordinates {(0.2, 0.0493522) (0.4, 0.123377) (0.6, 0.222817)  (0.8, 0.35421) (1.0, 0.528832) };
\end{axis}
\end{tikzpicture}
\caption{This figure evaluates PA$_{\text{(B-SB)}}$, combined with three power split schemes, SP$_{\text{(Uni)}}$, SP$_{\text{(FTPC)}}$, and SP$_{\text{(Opt)}}$. In SP$_{\text{(Opt)}}$, the power split is optimal for each pair. OMA$_{\text{(Opt)}}$ and NOMA$_{\text{(Opt)}}$ are baselines. }
\label{fig:Neighbored}
\end{figure}

 The asynchronous fixed-point iterations converge to the fixed point that is the same as obtained by its synchronized version. Intuitively, the fixed point is unique, regardless of how we reach it.

\section{Performance Evaluation}
\label{sec:simulation}

\begin{table}[!t]
\centering
\caption{Simulation Parameters.}
\begin{tabular}{ll}
\toprule
\textbf{Parameter} & \textbf{Value} \\
Cell radius & $500$ m\\
Carrier frequency & $2$ GHz \\
Total bandwidth & $20$ MHz\\
Cell load limit $\bar{\rho}$ & $1.0$ \\
Path loss model & COST-231-HATA \\
Shadowing (Log-normal) & $6$ dB standard deviation\\
Fading & Rayleigh flat fading \\
Noise power spectral density & $-173$ dBm/Hz \\
RB power $p_i$ ($i\in\I$) & $800$ mW \\
Convergence tolerance ($\epsilon$) & $10^{-4}$ \\
\bottomrule
\end{tabular}
\label{tab:sim}	
\end{table}

We use a cellular network of $19$ cells. To eliminate edge effects, wrap-around technique~\cite{wrap-around} is applied. Inside each cell, $30$ UEs are randomly and uniformly distributed. In each cell, there are in total ${30 \choose 2}=435$ possible choices for user pairing in NOMA. User demands are set to be a uniform value $d$. In the simulations, $d$ is normalized by $M\times B$ in~\eqref{eq:cj} and~\eqref{eq:cju}, and belongs to $(0,1]$. The network in OMA reaches the resource limit at $d=1.0$, i.e., any $d>1.0$ leads to at least one cell being overload in OMA.
 Other parameters are given in~\tablename~\ref{tab:sim}.

We consider two objectives for performance evaluation: resource efficiency and load balancing. For resource efficiency, the objective function is $F(\bm{\rho})=\sum_{i\in\I} \rho_i$, i.e., to minimize the total network time-frequency resource consumption (or cells' average resource consumption if divided by $n$). For load balancing, we adopt min-max fairness and the objective function is $F(\bm{\rho})=\max_{i\in\I} \rho_i$. Section~\ref{subsec:sim-power-alloc} and Section~\ref{subsec:sim-pairing} provide results for power allocation and user pairing, respectively. 
The optimal OMA, named OMA$_{\text{(Opt)}}$, is obtained by fixing $\vec{y}$ to $\vec{0}$ in \minf~and solving the remaining problem to optimality\footnote{With $\vec{y}$ being fixed to $\vec{0}$ in~\minf, the variables $\vec{q}$ and $\vec{x}$ disappear. Then we modify Line~\ref{alg:network-single-cell} of \textsc{M-Cell} to be ``$\rho_i^{(k)}=\sum_{j\in\J_i}d_j/c_j(\bm{\rho}_{-i})$'' and Line~\ref{alg:network-return} to be ``$\textbf{return}~\bm{\rho}^{(k)}$''. The modified \textsc{M-Cell} gives the optimal load for OMA (see~\cite{6204009} for further details).}. 
The proposed optimal NOMA solution is named NOMA$_{\text{(Opt)}}$ in the remaining context.

\subsection{Power Allocation}
\label{subsec:sim-power-alloc}

We use OMA$_{\text{(Opt)}}$ as baseline. As for NOMA, the pairing candidate set $\U$ initially covers all pairs of UEs in each cell. Then, those pairs not fulfilling Lemma~\ref{rmk:decoding} are dropped from $\U$. We then use \textsc{M-Cell} to compute NOMA$_{\text{(Opt)}}$.
Besides the optimal NOMA, we implement two other sub-optimal NOMA power split schemes for comparison. One is named ``SP$_{\text{(Uni)}}$'', in which the power $p_i$ splits equally between $q_{\oplus\u}$ and $q_{\ominus\u}$ for any pair $\u=\{\oplus,\ominus\}$ ($\u\in\U$). The other is ``fractional transmit power control'' (FTPC), named SP$_{\text{(FTPC)}}$, using a parameter to control the fairness for power split. We set this parameter to be $0.4$ as recommended in \cite{6666209}. Under both SP$_{\text{(Uni)}}$ and SP$_{\text{(FTPC)}}$, we use the method in Section~\ref{subsec:cell} to compute the optimal pair selection. Both two power split schemes are easily accommodated by \textsc{M-Cell}. 

\figurename~\ref{fig:load} shows the total load and the maximum load in function of normalized demand. As expected, the cell load levels monotonically increase with user demand. At high user demand, NOMA$_{\text{(Opt)}}$ dramatically improves the load performance. For $d=1.0$, it achieves 31\% better performance than OMA$_{\text{(Opt)}}$ for both total load and maximum load. The two sub-optimal solutions SP$_{\text{(Uni)}}$ and SP$_{\text{(FTPC)}}$ also result in load improvement than OMA$_{\text{(Opt)}}$. 
Compared to the two sub-optimal solutions, the improvement achieved by NOMA$_{\text{(Opt)}}$ over OMA$_{\text{(Opt)}}$ is doubled or more. On average, by using the same amount of time-frequency resource, NOMA$_{\text{(Opt)}}$ delivers 33\% more bits demand than OMA$_{\text{(Opt)}}$. Besides, SP$_{\text{(FTPC)}}$ achieves better performance than SP$_{\text{(Uni)}}$, as the former takes into account the channel conditions in power split. Generally, in SP$_{\text{(FTPC)}}$, UE with worse channel is allocated with more power.

In summary, 
power allocation has considerably large influence on NOMA. 
Even if the UE pairs are optimally selected, sub-optimal power allocations in NOMA have significant deviation from optimal NOMA.

\subsection{User Pairing}
\label{subsec:sim-pairing}

\pgfplotsset{compat=1.12,
        /pgfplots/ybar legend/.style={
        /pgfplots/legend image code/.code={%
        \draw[##1,/tikz/.cd,bar width=3pt,yshift=-0.2em,bar shift=0pt]
                plot coordinates {(0cm,0.8em)};},
},
}
\begin{figure}[!h]  
\begin{tikzpicture}
\begin{axis}[
	ymode = log,
	xlabel={Iteration $k$ in \textsc{M-Cell}},
	ylabel={$\norm{\bm{\rho}^{(k)}-\bm{\rho}^{(k-1)}}_{\infty}$},
	label style = {font=\fontsize{9pt}{10pt}\selectfont},
	legend cell align={left},
	legend pos = north east,
	legend style = {font=\fontsize{8pt}{10pt}\selectfont},
	axis background/.style={fill=white},
minor x tick num=0,
minor y tick num=4,
major tick length=0.15cm,
minor tick length=0.075cm,
xtick = {2,3,4,...,20},
tick style={semithick,color=black},
	height=0.667\linewidth,
	width=0.97\linewidth,
	xmin=2,
	xmax=10,
]
	
\addplot [smooth,dashed,color=black] coordinates {
	(2, 0.0577461)
	(3, 0.00203188)
	(4, 0.0000852924)
	(5, 3.65109*10^-6)
	(6, 1.36205*10^-7)
	(7, 6.07099*10^-9)
	(8, 2.07329*10^-10)
	(9, 8.0064*10^-12)
	(10, 3.09183*10^-13)
	(11, 1.19397*10^-14)
	(12, 4.61075*10^-16)
	(13, 1.78053*10^-17)
	(14, 6.87588*10^-19)
	(15, 2.65526*10^-20)
	(16, 1.02538*10^-21)
	(17, 3.95971*10^-23)
	(18, 1.52912*10^-24)
	(19, 5.90499*10^-26)
	(20, 2.28033*10^-27)
};

\addplot [smooth, color=brown] coordinates {
	(2, 0.0536244)
	(3, 0.00215018)
	(4, 0.0000852924)
	(5, 3.68565*10^-6)
	(6, 1.52194*10^-7)
	(7, 5.81971*10^-9)
	(8, 1.92194*10^-10)
	(9, 7.14749*10^-12)
	(10, 2.65807*10^-13)
	(11, 9.88507*10^-15)
	(12, 3.67615*10^-16)
	(13, 1.36712*10^-17)
	(14, 5.08417*10^-19)
	(15, 1.89075*10^-20)
	(16, 7.03147*10^-22)
	(17, 2.61493*10^-23)
	(18, 9.72462*10^-25)
	(19, 3.61648*10^-26)
	(20, 1.34493*10^-27)
};

\addplot [smooth, color=blue, dash dot] coordinates {
	(2, 0.0496822)
	(3, 0.00239899)
	(4, 0.000124095)
	(5, 6.43253*10^-6)
	(6, 3.34073*10^-7)
	(7, 1.73539*10^-8)
	(8, 5.88846*10^-10)
	(9, 2.60267*10^-11)
	(10, 1.15037*10^-12)
	(11, 5.08458*10^-14)
	(12, 2.24736*10^-15)
	(13, 9.93322*10^-17)
	(14, 4.39044*10^-18)
	(15, 1.94055*10^-19)
	(16, 8.57716*10^-21)
	(17, 3.79106*10^-22)
	(18, 1.67563*10^-23)
	(19, 7.40622*10^-25)
	(20, 3.27351*10^-26)
};

\addplot [smooth, color=red, dotted,thick] coordinates {
	(2, 0.0241175)
	(3, 0.00151652)
	(4, 0.000109303)
	(5, 8.59927*10^-6)
	(6, 5.46573*10^-7)
	(7, 2.7344*10^-8)
	(8, 1.41131*10^-9)
	(9, 7.62465*10^-11)
	(10, 5.11924*10^-12)
	(11, 2.22543*10^-13)
	(12, 1.2023*10^-14)
	(13, 6.49545*10^-16)
	(14, 3.50919*10^-17)
	(15, 1.89585*10^-18)
	(16, 1.02424*10^-19)
	(17, 5.53348*10^-21)
	(18, 2.98948*10^-22)
	(19, 1.61508*10^-23)
	(20, 8.72551*10^-25)
};
\legend{$d=1.0$, $d=0.7$, $d=0.4$, $d=0.1$}
\end{axis}
\end{tikzpicture}
\caption{This figure shows the norm $\norm{\cdot}_{\infty}$ in function of iteration $k$ in \textsc{M-Cell}, under the uniform demands $0.1$, $0.4$, $0.7$, and $1.0$, respectively.}
\label{fig:convergence}
\end{figure}
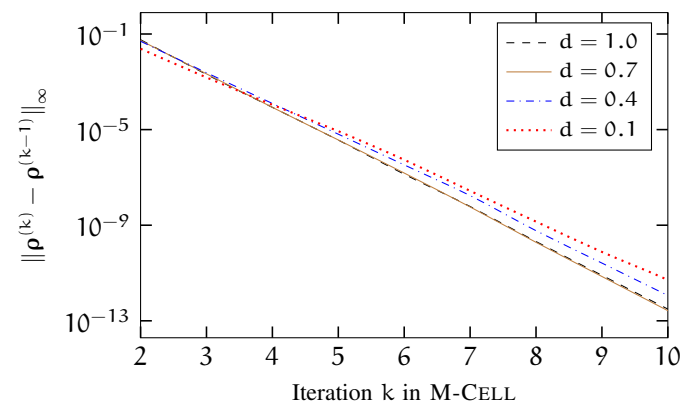

We study the influence of user pairing by considering
two sub-optimal ones~\cite{7273963}, named ``PA$_{\text{(B-W)}}$'' and ``PA$_{\text{(B-SB)}}$'', respectively. Suppose we sort the UEs in descending order of their channel conditions. In PA$_{\text{(B-W)}}$, the UE with the best channel condition is paired with the UE with the worst, and the UE with the second best is paired with one with the second worst, and so on. In PA$_{\text{(B-SB)}}$, the UE with the best channel condition is paired with the one with the second best, and so on. See~\figurename{s}~\ref{fig:pairing-best-worst} and \ref{fig:pairing-neighbored} for an illustration. In addition, we examine to what extend pair filtering (by Lemma~\ref{rmk:decoding}) affects performance. 
For filtered $\U$, optimal pair selection is done by Section~\ref{subsec:cell}. For non-filtered $\U$, we apply \textsc{M-Cell} even though there is no theoretical guarantee on optimality. Convergence, however, is observed for all the instances we considered.  
\figurename{s}~\ref{fig:pairing-filtered} and \ref{fig:pairing-non-filtered} illustrated the resulted selection patterns.

In~\figurename~\ref{fig:cell-load}, we show the load levels of all $19$ cells with $d=1.0$, under both filtered and non-filtered $\U$. In this specific scenario, $|\U|$ is reduced from ${30 \choose 2}\times 19=8265$ to $5779$ after being filtered by Lemma~\ref{rmk:decoding}. We choose $d=1.0$ because the performance difference among the solutions is the largest. 
There is very slight difference in cell load levels between the two cases. Numerically, the differences between them are only $0.1\%$ and $0.5\%$ for average and maximum cell load, respectively. This result is coherent with \figurename~\ref{fig:pairing-filtered} and \figurename~\ref{fig:pairing-non-filtered}. One can see that the patterns of the two pair selection solutions are almost identical. Thus, pair filtering by Lemma~\ref{rmk:decoding} is effective in reducing the number of candidate pairs, with virtually no impact on performance.

In \figurename~\ref{fig:Best-Worst} and \figurename~\ref{fig:Neighbored}, we respectively evaluate PA$_{\text{(B-W)}}$ and PA$_{\text{(B-SB)}}$, combined with three power split schemes SP$_{\text{(Uni)}}$, SP$_{\text{(FTPC)}}$, and SP$_{\text{(Opt)}}$. 
In SP$_{\text{(Opt)}}$, we use the algorithm~\textsc{Split} to compute the optimal power split for each pair.  All of SP$_{\text{(Uni)}}$, SP$_{\text{(FTPC)}}$, and SP$_{\text{(Opt)}}$ are put into the framework of \textsc{M-Cell} but with fixed pair selection PA$_{\text{(B-W)}}$ or PA$_{\text{(B-SB)}}$. In addition, 
OMA$_{\text{(Opt)}}$ and NOMA$_{\text{(Opt)}}$ are also included for comparison as baselines. 

One can see that all the NOMA schemes outperform OMA$_{\text{(Opt)}}$. In \figurename~\ref{fig:Best-Worst}, with PA$_{\text{(B-W)}}$, SP$_{\text{(FTPC)}}$ outperforms SP$_{\text{(Uni)}}$. SP$_{\text{(Opt)}}$ beats the other two. On one hand, there is non-negligible gap in load performance between SP$_{\text{(Opt)}}$ and NOMA$_{\text{(Opt)}}$, even though in SP$_{\text{(Opt)}}$, the power split is optimal for the PA$_{\text{(B-W)}}$ pairing. Hence pair selection plays an important role for NOMA performance. On the other hand, SP$_{\text{(Opt)}}$ yields significantly load improvement compared to OMA$_{\text{(Opt)}}$, and we conclude that PA$_{\text{(B-W)}}$ is a good sub-optimal pair selection for NOMA. Indeed, PA$_{\text{(B-W)}}$ pairs the UEs in a greedy way, aiming at maximizing the diversity of channel conditions of paired UEs. 
As shown in~\figurename~\ref{fig:pairing-filtered}, the optimal pair selection has a similar trend. 
The difference is that optimal pairing has a more ``global view'' than PA$_{\text{(B-W)}}$. 
In \figurename~\ref{fig:Neighbored}, under PA$_{\text{(B-SB)}}$, SP$_{\text{(Uni)}}$, SP$_{\text{(FTPC)}}$, and SP$_{\text{(Opt)}}$ improve the load very slightly. All of the three are far from the global optimum and the gap is large under high user demands. We conclude that PA$_{\text{(B-SB)}}$ is not as effective as PA$_{\text{(B-W)}}$ in terms of network load optimization. 

As the overall conclusion, jointly optimizing power allocation and user pairing is important for the performance of NOMA.


\subsection{Convergence Analysis}
\label{subsec:sim-convergence}

We show the convergence performance of \textsc{M-Cell} in 
\figurename~\ref{fig:convergence}, for demands $0.3$, $0.5$, $0.7$, and $1.0$, respectively. 
Initially, $\rho^{(0)}_i=1$ ($i\in\I$). 
We observe that \textsc{M-Cell} converges very fast. With higher demand, the convergence becomes slightly faster. High accuracy is reached after a very few iterations. 
For all the demands consider in the figure, even if we terminate \textsc{M-Cell} after a very few iterations, the obtained solution is close to the optimum.

\section{Conclusions}
\label{sec:conclusion}

This paper has investigated optimal resource management in multi-cell NOMA, with power allocation and user pairing being considered simultaneously. Joint optimization of both is shown to be very important for NOMA performance. The proposed system model admits a mixed use of OMA and NOMA for the users. Therefore, network architectures that support various multiple access techniques can be analyzed under this model. Finally, as for future work, the paper suggests that mathematical tools in SIF are useful for analyzing multi-cell NOMA. In summary, NOMA is a promising technique for spectrum efficiency enhancement and cell load balancing.

\section*{Acknowledgement}
The work has been partially supported by the Luxembourg National Research Fund (FNR) CORE project ROSETTA (C17/IS/11632107).

\begin{strip}
\appendix
We remark that the function~\eqref{eq:conv} is first referred to in Section~\ref{subsec:pair} and is used in the formulation~\eqref{eq:min_pair2}. The two functions~\eqref{eq:H+} and~\eqref{eq:H-} are first referred to in Section~\ref{subsec:pair} and are used in the algorithm \textsc{Split}.
\begin{equation}
\text{Cv$_{\u}$}(c_{\oplus\u},c_{\ominus\u},\bm{\rho}_{-i})=\log\left[\frac{w_{\oplus}(\bm{\rho}_{-i})e^{(c_{\oplus\u}+c_{\ominus\u})}+(w_{\ominus}(\bm{\rho}_{-i})-w_{\oplus}(\bm{\rho}_{-i}))e^{c_{\ominus\u}}}{p_i+w_{\ominus}(\bm{\rho}_{-i})}\right].	
\label{eq:conv}
\end{equation}	
\begin{equation}
H_{\oplus\u}\left(x_{\u},\bm{\rho}_{-i}\right)=\left\{
\begin{array}{ll}
d_{\oplus}/x_{\u}  & c^K_{\oplus\u}x_{\u}> d_{\oplus}\\
\log\left[
\frac{(p_i+w_{\ominus}(\bm{\rho}_{-i}))/e^{{d_{\ominus}}/{x_{\u}}}+w_{\oplus}(\bm{\rho}_{-i})-w_{\ominus}(\bm{\rho}_{-i})}{w_{\oplus}(\bm{\rho}_{-i})}
\right]
& \text{Otherwise}
\end{array}
\right.
\label{eq:H+}
\end{equation}
\begin{equation}
H_{\ominus\u}\left(x_{\u},\bm{\rho}_{-i}\right)=\left\{
\begin{array}{ll}
d_{\ominus}/x_{\u}	 & c^K_{\ominus\u}x_{\u}> d_{\ominus} \\
\log\left[
\frac{p_i+w_{\ominus}(\bm{\rho}_{-i})}
{w_{\oplus}(\bm{\rho}_{-i})e^{{d_{\oplus}}/{x_{\u}}}+w_{\ominus}(\bm{\rho}_{-i})-w_{\oplus}(\bm{\rho}_{-i})}
\right] 
 & \text{Otherwise} 
\end{array}
\right.
\label{eq:H-}
\end{equation}
\end{strip}

\bibliographystyle{IEEEtran}
\bibliography{ref.bib}

\end{document}